\documentclass[conference]{IEEEtran}
\IEEEoverridecommandlockouts

\usepackage{amsmath,amssymb,amsthm,mathtools}
\usepackage{bm}
\usepackage[normalem]{ulem}
\usepackage{booktabs}
\usepackage{cite}
\usepackage{hyperref}
\hypersetup{hidelinks}
\usepackage{graphicx}
\usepackage{tikz}
\usetikzlibrary{arrows.meta,calc,positioning}

\newif\ifextendedversion
\ifdefined\ARXIVEXTENDED
  \extendedversiontrue
\else
  \extendedversionfalse
\fi

\newtheorem{theorem}{Theorem}

\newtheorem{corollary}{Corollary}
\theoremstyle{remark}

\newcommand{\Fq}{\mathbb{F}_q}
\newcommand{\Ftwo}{\mathbb{F}_2}
\newcommand{\mbC}{\mathbf{C}}

\newcommand{\Msmall}{\frac{1}{2N-1}}
\newcommand{\Meven}{\frac{2}{3N-2}}
\newcommand{\Mfull}{\frac{N-1}{N}}

\newcommand{\avgrank}{\widehat{\operatorname{rank}}}

\renewcommand{\Im}{\mathrm{Im}}
\allowdisplaybreaks[1]
\begin{document}
\raggedbottom
\setlength{\textfloatsep}{12pt plus 2pt minus 2pt}

\title{Cache-Aided Linear Function Computation}

\author{\IEEEauthorblockN{Xiaohong Cai, Ruze Zhang, Laigang Guo, and Raymond W. Yeung}
\thanks{Xiaohong Cai and Raymond W. Yeung are with the Department of Information Engineering, The Chinese University of Hong Kong, Hong Kong SAR, China (e-mail: cx021@ie.cuhk.edu.hk; whyeung@ie.cuhk.edu.hk).}
\thanks{Ruze Zhang is with the School of Mathematical Sciences and the Key Laboratory of Pure Mathematics and Combinatorics (LPMC), Nankai University, Tianjin 300071, China (e-mail: rzzhang@nankai.edu.cn).}
\thanks{Laigang Guo is with the Laboratory of Mathematics and Complex Systems (Ministry of Education), School of Mathematical Sciences, Beijing Normal University, Beijing 100875, China (e-mail: lgguo@bnu.edu.cn).}
}

\maketitle

\begin{abstract}
We study cache-aided linear computation in distributed systems where each source generates and stores its own data. An external receiver requests a linear combination of these data, and the sources use information cached before the request is known to reduce the communication needed to recover the requested result. In our model, each of the $N$ sources stores its own independent file over a finite field and has an additional cache of normalized size $M$. The receiver must recover any requested linear combination with zero error from the messages sent by the sources over separate noiseless links. We establish a lower bound on the worst-case communication rate over any finite field by analyzing several requests under the same cache placement. This bound determines the optimal rate when no additional cache is available. For binary files, we propose parity-aided placement (PAP), a coded-placement scheme that combines cached parities with uncoded symbols. For every $N\geq3$, we characterize its communication rate for each request and optimize its placement parameters to obtain a closed-form worst-case rate. The resulting achievable rate meets the lower bound for $0\leq M\leq1/(2N-1)$ and $M\geq(N-1)/N$, establishing the optimal memory--rate tradeoff in these two ranges.
\end{abstract}


\section{Introduction}

In distributed systems with locally generated data, different source nodes retain their own data, while an external receiver may request a function of the data held by multiple sources. For example, Earth observation satellites can share observation data for processing in orbit~\cite{ValenteEramoLavacca2023}, and multiple aerial sensing nodes can send their observations to a ground station for data fusion through an unmanned aerial vehicle (UAV) relay~\cite{XingWangTengLi2022}. When the receiver needs only a function of the source data, a small result does not necessarily mean that little communication is required. The receiver may still need separate transmissions from many sources, whose total size can far exceed that of the result. Exchanging and caching information before the requested function is specified can change the information available at each source and reduce the communication needed for subsequent computation. The challenge is to select cache contents that support different future requests without knowing which function will be requested.

Caching can reduce the transmission load by making information available before requests arrive. 
In the classical coded-caching model, a server holds the entire file library, and receivers cache information to reduce the communication required for subsequent file requests over a shared link~\cite{MaddahAliNiesen2014,YuMaddahAliAvestimehr2018}. Subsequent work established bounds for general placement~\cite{YuMaddahAliAvestimehr2019} and developed coded placements that cache linear combinations of file segments~\cite{TianChen2018}. Cache-aided scalar and general linear function retrieval extend this setting to requests for linear functions of the library files, with the entire library still held by a central server and caches located at the receivers~\cite{WanSunJiTuninettiCaire2021,WanSunJiTuninettiCaire2021GLFR}.
Device-to-device caching enables users to retrieve files from one another's caches~\cite{JiCaireMolisch2016,YaparWanSchaeferCaire2019}. Coding can also reduce communication in distributed computing~\cite{LiMaddahAliYuAvestimehr2018,WanSunJiCaire2022}. 

Network function computation studies the recovery of functions of distributed source messages over a network. Previous work established cut-set bounds for network function computation~\cite{AppuswamyFranceschettiKaramchandaniZeger2011}, algebraic conditions for computing linear functions~\cite{AppuswamyFranceschetti2014}, and improved capacity bounds for general target functions and network topologies~\cite{GuangYeungYangLi2019}. For the binary arithmetic sum, the zero-error compression capacity has been characterized under different source-observation patterns at two encoders~\cite{GuangZhang2023}.

 We study how caching at the sources can reduce the communication required for linear function computation. In our model, each of the $N$ sources retains its own file and has an additional cache of normalized size $M$. The caching functions are fixed before the request is revealed and must support every linear combination requested by an external receiver. We seek to characterize the tradeoff between cache size and the worst-case total communication required for zero-error recovery. Our main results are as follows.

We establish a converse bound for an arbitrary finite field. For any $m\times N$ query matrix $A$, the bound depends on the average rank of its $m\times n$ column submatrices. Maximizing this quantity over $A$ for each pair $(m,n)$ gives
\begin{equation}
R(\Fq,N,M)\geq
\max\nolimits_{1\leq m,\,n\leq N}\frac{N}{mn}
\bigl(\avgrank_{m,n,N}-(N-n)M\bigr).
\label{eq:main-multi-query}
\end{equation}

For the binary case, we introduce \emph{parity-aided placement} (PAP), which combines cached parities with uncoded symbols. For every $N\geq3$, we derive its communication rate for each request support size and optimize the placement parameters to obtain a closed-form worst-case rate, asymptotically achievable for every $M$. In the small-cache regime, PAP meets the $m=n=1$ instance of~\eqref{eq:main-multi-query} and gives
\begin{equation}
R(\Ftwo,N,M)=N-N(N-1)M,
\quad 0\leq M\leq\frac{1}{2N-1}. 
\end{equation}
PAP also achieves the optimal rate $R(\Ftwo,N,M)=1$ for $M\geq(N-1)/N$, thereby determining the optimal memory--rate tradeoff in these two regimes.

\section{Model of Cache-Aided Linear Computation}

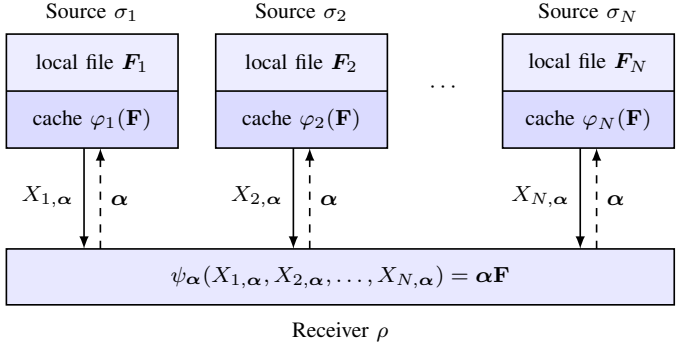
\begin{figure}[!t]
\centering
\resizebox{\columnwidth}{!}{%
\begin{tikzpicture}[
    >=latex,
    font=\footnotesize\rmfamily,
    every node/.style={font=\footnotesize\rmfamily},
    sourcefile/.style={fill=blue!7,rectangle,align=center,minimum width=2.20cm,minimum height=0.72cm,inner xsep=1pt,inner ysep=0pt},
    sourcecache/.style={fill=blue!14,rectangle,align=center,minimum width=2.20cm,minimum height=0.72cm,inner xsep=1pt,inner ysep=0pt},
    receiver/.style={fill=blue!9,draw,semithick},
    message/.style={fill=white,inner sep=1.5pt}
]
\node[sourcefile] (s1-file) at (-3.20,2.10) {local file $\boldsymbol F_1$};
\node[sourcecache,anchor=north] (s1-cache) at (s1-file.south) {cache $\varphi_1(\mathbf F)$};
\draw[semithick] (s1-file.north west) rectangle (s1-cache.south east);
\draw[semithick] (s1-file.south west) -- (s1-file.south east);
\node at ([yshift=0.28cm]s1-file.north) {Source $\sigma_1$};

\node[sourcefile] (s2-file) at (-0.50,2.10) {local file $\boldsymbol F_2$};
\node[sourcecache,anchor=north] (s2-cache) at (s2-file.south) {cache $\varphi_2(\mathbf F)$};
\draw[semithick] (s2-file.north west) rectangle (s2-cache.south east);
\draw[semithick] (s2-file.south west) -- (s2-file.south east);
\node at ([yshift=0.28cm]s2-file.north) {Source $\sigma_2$};

\node[sourcefile] (sN-file) at (3.20,2.10) {local file $\boldsymbol F_N$};
\node[sourcecache,anchor=north] (sN-cache) at (sN-file.south) {cache $\varphi_N(\mathbf F)$};
\draw[semithick] (sN-file.north west) rectangle (sN-cache.south east);
\draw[semithick] (sN-file.south west) -- (sN-file.south east);
\node at ([yshift=0.28cm]sN-file.north) {Source $\sigma_N$};
\node at (1.35,1.78) {$\cdots$};

\coordinate (receiver-top) at (0,-0.30);
\coordinate (receiver-bottom) at (0,-1.02);
\coordinate (receiver-nw) at (s1-file.west |- receiver-top);
\coordinate (receiver-ne) at (sN-file.east |- receiver-top);
\coordinate (receiver-sw) at (s1-file.west |- receiver-bottom);
\coordinate (receiver-se) at (sN-file.east |- receiver-bottom);
\draw[receiver] (receiver-sw) rectangle (receiver-ne);
\node at ($(receiver-sw)!0.5!(receiver-ne)$) {$\psi_{\boldsymbol\alpha}(X_{1,\boldsymbol\alpha},X_{2,\boldsymbol\alpha},\ldots,X_{N,\boldsymbol\alpha})=\boldsymbol\alpha\mathbf F$};
\node[below=0.10cm] at ($(receiver-sw)!0.5!(receiver-se)$) {Receiver $\rho$};

\coordinate (r1) at (s1-cache.center |- receiver-top);
\coordinate (r2) at (s2-cache.center |- receiver-top);
\coordinate (rN) at (sN-cache.center |- receiver-top);
\draw[->,semithick] ([xshift=-3pt]s1-cache.south) -- node[message,left=2pt,pos=0.48] {$X_{1,\boldsymbol\alpha}$} ([xshift=-3pt]r1);
\draw[->,semithick,dashed] ([xshift=3pt]r1) -- node[message,right=2pt,pos=0.48] {$\boldsymbol\alpha$} ([xshift=3pt]s1-cache.south);
\draw[->,semithick] ([xshift=-3pt]s2-cache.south) -- node[message,left=2pt,pos=0.48] {$X_{2,\boldsymbol\alpha}$} ([xshift=-3pt]r2);
\draw[->,semithick,dashed] ([xshift=3pt]r2) -- node[message,right=2pt,pos=0.48] {$\boldsymbol\alpha$} ([xshift=3pt]s2-cache.south);
\draw[->,semithick] ([xshift=-3pt]sN-cache.south) -- node[message,left=2pt,pos=0.48] {$X_{N,\boldsymbol\alpha}$} ([xshift=-3pt]rN);
\draw[->,semithick,dashed] ([xshift=3pt]rN) -- node[message,right=2pt,pos=0.48] {$\boldsymbol\alpha$} ([xshift=3pt]sN-cache.south);
\end{tikzpicture}%
}
\caption{The model of cache-aided linear computation $(\Fq,N,M)$.}
\label{fig:cache-aided-linear-computation}
\end{figure}

We consider a \emph{cache-aided linear computation model} $(\mathbb{F}_q,N, M)$ as depicted in Fig.\,\ref{fig:cache-aided-linear-computation}.
In the model, there are $N$  \emph{sources} $\sigma_1,\sigma_2,\cdots,\sigma_N$ and a  \emph{receiver} $\rho$, where each source is connected to the receiver by a noiseless link.
For each $i\in[N]\triangleq\{1,2,\cdots,N\}$, let the \emph{file} generated by the source~$\sigma_i$ be a  random $L$-dimensional row vector $\boldsymbol{F}_i=(F_{i,1},F_{i,2},\cdots,F_{i,L})$ uniformly distributed over   $\mathbb{F}_q^{L}$,\footnote{In this paper, for any positive integer $Z$, we use $\mathbb{F}_q^{Z}$ to denote the set of all $Z$-dimensional row vectors over $\mathbb{F}_q$.}
where $\mathbb{F}_q$ is a finite field of size $q$ and $L$ is a positive integer  referred to as the \emph{file length}.
All files $\boldsymbol{F}_1,\boldsymbol{F}_2,\cdots,\boldsymbol{F}_N$
are assumed to be  mutually independent.
Further, let
$
\mathbf{F}\triangleq
\big[\boldsymbol{F}_1^{\mathsf T}\ \boldsymbol{F}_2^{\mathsf T}\ \cdots\ \boldsymbol{F}_N^{\mathsf T}\big]^{\mathsf T},$
which is referred to as the \emph{file matrix}.

Furthermore, due to prior data exchanges, prefetching, or overhearing via the broadcast nature of the wireless channel, the sources have access to some side information about other files.
More precisely, for each $i\in [N]$,   the source~$\sigma_i$  can store an arbitrary (possibly nonlinear) function of the file matrix $\mathbf{F}$  as the \emph{cached content}, and the cached content takes at most $q^{ML}$ possible values.
Accordingly, we refer to $ML$ as the \emph{cache size}.
Here, $M$ is a nonnegative real number called the \emph{normalized cache size}.
The caching functions are fixed before the coefficient vector $\boldsymbol{\alpha}$ is revealed.
For each  \emph{coefficient vector} $\boldsymbol{\alpha}\triangleq(\alpha_1, \alpha_2, \cdots, \alpha_N)  \in \mathbb{F}_q^N$,
 the receiver $\rho$ is required to compute the linear function
$\boldsymbol{\alpha}\mathbf{F}=\sum_{i=1}^N \alpha_i \boldsymbol{F}_i
$
with zero error through the following steps:
\begin{itemize}
    \item[1)] the receiver $\rho$ first broadcasts the coefficient vector $\boldsymbol{\alpha}$ to all source nodes;
    \item[2)] each source node $\sigma_i$ generates and transmits an encoded message, say  $X_{i,\boldsymbol{\alpha}}$,
    based on its local file, its cached content  and the coefficient vector $\boldsymbol{\alpha}$;
    \item[3)] after receiving all transmitted messages, the receiver $\rho$ computes $\boldsymbol{\alpha}\mathbf{F}$ with zero error.
\end{itemize}

Now, we formally
define an  \emph{$L$-dimensional (linear-computation) code}  for the model $(\mathbb{F}_q,N, M)$  such that for each  coefficient vector
$\boldsymbol{\alpha}\in \mathbb{F}_q^N$, the receiver $\rho$ can compute  the linear function $\boldsymbol{\alpha}\mathbf{F}$ of the $L$-dimensional files $\boldsymbol{F}_1,\boldsymbol{F}_2,\cdots,\boldsymbol{F}_N$
with zero error.
To be specific, an $L$-dimensional code  for the model  $(\mathbb{F}_q,N, M)$   consists of
\begin{itemize}
\item a \emph{caching function} $\varphi_i$ for generating the cache content at every source node $\sigma_i$
$\varphi_i:~ \mathbb{F}_q^{N \times L} \to \Im\,\varphi_i$
such that the cached content takes at most $q^{ML}$ possible values,
i.e.,  $|\Im\,\varphi_i|\leq q^{ML}$;\footnote{In the   paper, we use $\Im\,\varphi$  to stand for the image set of a function $\varphi$.}

\item an \emph{encoding function} $\phi_{i,\boldsymbol{\alpha}}$ for each $\boldsymbol{\alpha} \in \mathbb{F}_q^N$ at every source node $\sigma_i$
$\phi_{i,\boldsymbol{\alpha}}:~ \mathbb{F}_q^{L}\times \Im\,\varphi_i \to \Im\,\phi_{i,\boldsymbol{\alpha}},$
which is used to generate a message transmitted from the source node $\sigma_i$ to the receiver $\rho$ based on the local file $\boldsymbol{F}_i$, the cached content $\varphi_i(\mathbf{F})$ and the coefficient vector~$\boldsymbol{\alpha}$;

\item a \emph{decoding function} $\psi_{\boldsymbol{\alpha}}$ for each $\boldsymbol{\alpha} \in \mathbb{F}_q^N$ at the receiver $\rho$, given by
$\psi_{\boldsymbol{\alpha}}:~ \prod_{i=1}^N \Im\,\phi_{i,\boldsymbol{\alpha}} \to \mathbb{F}_q^L,$
which is used to compute the linear function
$ \boldsymbol{\alpha}\mathbf{F} $
based on the messages received from all the $N$ sources.
\end{itemize}
We write this code as $\mbC=\big\{\varphi_i,\phi_{i,\boldsymbol{\alpha}},\psi_{\boldsymbol{\alpha}}:~i\in [N],  \boldsymbol{\alpha}\in\mathbb{F}_q^N\big\},$ and it
is called \emph{admissible} if
\begin{align*}
\psi_{\boldsymbol{\alpha}}\Big( \phi_{i,\boldsymbol{\alpha}}\big( \boldsymbol{F}_i ,\,\varphi_i(\mathbf{F})\big):~i\in [N] \Big)= \boldsymbol{\alpha}\mathbf{F} ,\qquad\forall~\boldsymbol{\alpha}\in \mathbb{F}_q^{N}.
\end{align*}
For the  code $\mbC$, we define
 \begin{equation}\label{def:compr-pair-R_1-R_2}
R_{\boldsymbol{\alpha}}(\mbC)\triangleq \frac{1}{\,L\,}\cdot
\sum\nolimits_{i=1}^N  \log|\Im\,\phi_{i,\boldsymbol{\alpha}}|, \qquad\forall~\boldsymbol{\alpha} \in \mathbb{F}_q^N,
\end{equation}
where   ``$\log$'' denotes ``$\log_q$'', the logarithm with base $q$.
Then, $R_{\boldsymbol{\alpha}}(\mbC)$ is the total  communication cost, normalized by the file length, from all source nodes to the receiver $\rho$ by using the code $\mbC$ for computing the linear function $\boldsymbol{\alpha}\mathbf{F}$.
The \emph{(worst-case) communication rate} of the code $\mbC$ is defined by
\begin{equation}\label{def:rate}
R(\mbC) \triangleq \max\nolimits_{\boldsymbol{\alpha} \in \mathbb{F}_q^N} R_{\boldsymbol{\alpha}}(\mbC).
\end{equation}
Consequently, the \emph{optimal communication rate} for the model  $(\mathbb{F}_q,N,M)$  is defined by
\begin{align*}
 &R (\mathbb{F}_q,N,M) \triangleq \inf\Big\{R(\mbC):~\text{$\mbC$ is an admissible $L$-} \\ &\text{dimensional code, where $L$ is an arbitrary positive integer.}\Big\}.
\end{align*}
When the model  $(\mathbb{F}_q,N,M)$  is clear from the context, we simply write the optimal rate $ R (\mathbb{F}_q,N,M)$ as~$ R(M)$.

 \section{Lower Bound on the Optimal\\ Communication Rate}

In this section, we   prove a general lower bound on the optimal communication rate
$R(\mathbb{F}_q,N,M)$,
which is applicable for an arbitrary finite field $\mathbb{F}_q$, an arbitrary  number $N$ of sources, and an arbitrary normalized cache size $M$.
Fix $1\leq m,\,n\leq N$. For $A\in\Fq^{m\times N}$ and $\theta\subseteq[N]$ with $|\theta|=n$, let $A[\theta]$ be the submatrix formed by the columns indexed by $\theta$. Define
\begin{align}
\avgrank_{n}(A)
&=\frac{1}{\binom{N}{n}}
\sum\nolimits_{\substack{\theta\subseteq[N]\\|\theta|=n}}
\operatorname{rank}(A[\theta]),
\label{eq:average-rank}\\
\avgrank_{m,n,N}
&=\max\nolimits_{A\in\Fq^{m\times N}}\avgrank_{n}(A).
\label{eq:max-average-rank}
\end{align}
In particular, $\avgrank_{1,1,N}=1$ and $\avgrank_{N,N,N}=N$.

\begin{theorem}
For every finite field $\Fq$, $N\geq1$, and $M\geq0$,
\begin{equation*}
R(\Fq,N,M)\geq
\max\nolimits_{1\leq m,\,n\leq N}\frac{N}{mn}
\bigl(\avgrank_{m,n,N}-(N-n)M\bigr).
\end{equation*}
\label{thm:multi-query-converse}
\end{theorem}

\begin{proof}
See Appendix~\ref{app:converse-proof}.
\end{proof}

 % \begin{corollary}
% For every finite field $\Fq$ and every $N\geq1$,
% \[
% R(\Fq,N,0)=N.
% \]
% \label{cor:zero-cache}
% \end{corollary}

\section{Code Construction and Upper Bound on the Optimal Communication Rate}
\label{sec:PAP}\label{sec:pap}

In this section, we focus on the model $(\mathbb{F}_2,N,M)$ over the binary field $\mathbb{F}_2$. We construct a class of codes, referred to as
\emph{parity-aided-placement (PAP) codes}. With this code construction, we prove an
upper bound on the optimal communication rate $R(\mathbb{F}_2,N,M)$.

\subsection{Three Elementary Codes}
\label{subsec:PAP-three-codes}

We first present three elementary codes: an \emph{uncached code}, a \emph{full-placement code}, and an \emph{all-but-one parity (ABO) code}, which will be combined to construct a PAP code.

\smallskip
\noindent$\bullet$ \textbf{Uncached Code:}
Consider the   model
$(\mathbb{F}_2,N,M=0)$, where $N$ is an arbitrary positive integer and
$M=0$ indicates that no cached content is available at
any source node.
For each positive integer~$L$, we define  an      $L$-dimensional
\emph{uncached code}
$\mbC=\big\{\varphi_i,\phi_{i,\boldsymbol{\alpha}}, \psi_{\boldsymbol{\alpha}}: i\in[N],~\boldsymbol{\alpha}\in\mathbb{F}_2^N\big\}$
for the model $(\mathbb{F}_2,N,0)$   as follows:
\begin{itemize}

\item
no cached content is stored at any source node, i.e., the caching function
$\varphi_i$ at every source  $\sigma_i$ satisfies
$\varphi_i(\mathbf{F}) = \langle\,\rangle$ for $i\in[N],$
where $\langle\,\rangle$ denotes the empty vector;

\item the encoding function
$\phi_{i,\boldsymbol{\alpha}}$ for each coefficient vector
$\boldsymbol{\alpha}=(\alpha_1, \alpha_2, \cdots, \alpha_N)  \in\mathbb{F}_2^N$ at every source node
$\sigma_i$ is given by
$\phi_{i,\boldsymbol{\alpha}} \big( \boldsymbol{F}_i,\varphi_i(\mathbf{F}) \big) = \alpha_i\boldsymbol{F}_i;$

\item the  decoding function
$\psi_{\boldsymbol{\alpha}}$ for each
$\boldsymbol{\alpha}\in\mathbb{F}_2^N$ at the receiver $\rho$
is given by
\begin{align*}
\psi_{\boldsymbol{\alpha}}
\Big(
\phi_{i,\boldsymbol{\alpha}}
\big(
\boldsymbol{F}_i,\varphi_i(\mathbf{F})
\big):
i\in[N]
\Big)
=
\sum\nolimits_{i=1}^N
\alpha_i\boldsymbol{F}_i
=
\boldsymbol{\alpha}\mathbf{F}.
\end{align*}

\end{itemize}

For an arbitrary $i\in[N]$ and an arbitrary coefficient vector
$\boldsymbol{\alpha}=(\alpha_1, \alpha_2, \cdots, \alpha_N)  \in\mathbb{F}_2^N$, we can readily see that
\[
\log|\Im\,\phi_{i,\boldsymbol{\alpha}}| = \begin{cases} L, & \alpha_i\neq 0,\\ 0, & \alpha_i=0. \end{cases}
\]
This further implies that
\begin{align}
\sum\nolimits_{i=1}^N
\log|\Im\,\phi_{i,\boldsymbol{\alpha}}|&=
\big|\{i\in[N]:~\alpha_i\neq 0\}\big|L\nonumber\\
&=|\boldsymbol{\alpha}|L,
\label{un-code-R-alpha}
\end{align}
where $|\boldsymbol{\alpha}|$ denotes the Hamming weight of
$\boldsymbol{\alpha}$, i.e., the number of nonzero components of
$\boldsymbol{\alpha}$.

\smallskip
\noindent$\bullet$ \textbf{Full-Placement Code:}
Consider the model
$(\mathbb{F}_2,N,M=(N-1)/N)$, where $N$ is an arbitrary positive integer.
For each positive integer~$L$ such that $N$ divides $L$, we define an
$L$-dimensional \emph{full-placement code}
$
\mbC
=
\big\{
\varphi_i,
\phi_{i,\boldsymbol{\alpha}},
\psi_{\boldsymbol{\alpha}}:
i\in[N],~
\boldsymbol{\alpha}\in\mathbb{F}_2^N
\big\}
$
for the model $(\mathbb{F}_2,N,(N-1)/N)$ as follows.
First, we partition $[L]$ into $N$ disjoint subsets
$\mathcal{L}_1,\mathcal{L}_2,\cdots,\mathcal{L}_N$ of equal size, i.e.,
\begin{align}
[L]
=
\bigcup_{i=1}^N
\mathcal{L}_i,
\quad\text{and}\quad
|\mathcal{L}_i|
=
\frac{L}{N},
\quad
\forall~i\in[N].\label{eq:full-coordinate-partition}
\end{align}
\begin{itemize}

\item
the caching function $\varphi_i$ for generating the cache content at every
source node $\sigma_i$ is given by
\begin{align}
\varphi_i(\mathbf{F})
=
\Big(
F_{n,\ell}:
~n\in[N]\setminus\{i\},
~\ell\in\mathcal{L}_i
\Big),
\label{full-cache-function}
\end{align}
where we note that
$|\Im\,\varphi_i | = 2^{(N-1)|\mathcal{L}_i|} = 2^{(N-1)L/N};$

\item
the encoding function
$\phi_{i,\boldsymbol{\alpha}}$ for each coefficient vector
$\boldsymbol{\alpha}=(\alpha_1,\alpha_2,\cdots,\alpha_N)
\in\mathbb{F}_2^N$ at every source node $\sigma_i$ is given by
\begin{align}
\phi_{i,\boldsymbol{\alpha}}
\big(
\boldsymbol{F}_i,
\varphi_i(\mathbf{F})
\big)
=
\Bigg(
\sum\nolimits_{n=1}^N
\alpha_nF_{n,\ell}:
~\ell\in\mathcal{L}_i
\Bigg);
\label{full-encoding-function}
\end{align}

\item
the decoding function
$\psi_{\boldsymbol{\alpha}}$ for each
$\boldsymbol{\alpha}\in\mathbb{F}_2^N$ at the receiver $\rho$
is given by
\begin{align*}
&\psi_{\boldsymbol{\alpha}}
\Big(
\phi_{i,\boldsymbol{\alpha}}
\big(
\boldsymbol{F}_i,
\varphi_i(\mathbf{F})
\big):
i\in[N]
\Big)
\\
&=
\Bigg(
\sum\nolimits_{n=1}^N
\alpha_nF_{n,\ell}:
~\ell\in[L]
\Bigg)
=
\sum\nolimits_{n=1}^N
\alpha_n\boldsymbol{F}_n .
\end{align*}

\end{itemize}

We can readily see from \eqref{eq:full-coordinate-partition},
\eqref{full-encoding-function}, and the definition of the decoding
function that the code $\mbC$ is admissible.
Furthermore, for any coefficient vector
$\boldsymbol{\alpha}
=(\alpha_1,\alpha_2,\cdots,\alpha_N)\in\mathbb{F}_2^N$,
\begin{align}
\sum\nolimits_{i=1}^N
\log|\Im\,\phi_{i,\boldsymbol{\alpha}}|
=
\begin{cases}
0
&\quad\textup{if}~|\boldsymbol{\alpha}|=0,
\\
L
&\quad\textup{if}~|\boldsymbol{\alpha}|>0.
\end{cases}
\label{full-code-R-alpha}
\end{align}
To show this, when $|\boldsymbol{\alpha}|=0$, we have
$\boldsymbol{\alpha}=\boldsymbol{0}$, and hence
$\log|\Im\,\phi_{i,\boldsymbol{\alpha}}| = 0$ for all $i\in[N].$
This immediately implies that
$
\sum_{i=1}^N\log|\Im\,\phi_{i,\boldsymbol{\alpha}}|=0.
$
When $|\boldsymbol{\alpha}|>0$, by
\eqref{eq:full-coordinate-partition} and
\eqref{full-encoding-function}, we have
$\log|\Im\,\phi_{i,\boldsymbol{\alpha}}| = \frac{L}{N}$ for all $i\in[N].$
Therefore, we obtain that
$\sum_{i=1}^N \log|\Im\,\phi_{i,\boldsymbol{\alpha}}| = \sum_{i=1}^N \frac{L}{N} =L.$

\smallskip
\noindent$\bullet$ \textbf{All-But-One Parity Code:}
Consider the   model
$(\mathbb{F}_2,N,M=1/N)$, where $N\geq 2$ is an arbitrary positive integer.
For each positive integer~$L$ such that $N$ divides $L$, we define  an $L$-dimensional
\emph{all-but-one (ABO)  parity code}
$
\mbC
=
\big\{
\varphi_i,
\phi_{i,\boldsymbol{\alpha}},
\psi_{\boldsymbol{\alpha}}:
i\in[N],~
\boldsymbol{\alpha}\in\mathbb{F}_2^N
\big\}
$
for the model $(\mathbb{F}_2,N,1/N)$ as follows.
First, we partition $[L]$ into $N$ disjoint subsets $\mathcal{L}_1,\mathcal{L}_2,\cdots,\mathcal{L}_N$ of equal size as in~\eqref{eq:full-coordinate-partition}.
\begin{itemize}

\item
The  caching function  $\varphi_i$ for generating the cache content at every source node $\sigma_i$ is given by
\begin{align}
\varphi_i(\mathbf{F})
=
\Bigg(
\sum\nolimits_{n\in[N]\setminus\{i\}}
F_{n,\ell}:
~\ell\in\mathcal{L}_i
\Bigg),
\label{ABO-cache-function}
\end{align}
where we note that
$
\big|\Im\,\varphi_i\big|
=
2^{|\mathcal{L}_i|}
=
2^{L/N} $
by \eqref{eq:full-coordinate-partition};

\item
the encoding function
$\phi_{i,\boldsymbol{\alpha}}$ for each coefficient vector
$\boldsymbol{\alpha}=(\alpha_1,\alpha_2,\cdots,\alpha_N)
\in\mathbb{F}_2^N$ at every source node $\sigma_i$ is given by
\begin{align}
\phi_{i,\boldsymbol{\alpha}}
\big(
\boldsymbol{F}_i,
\varphi_i (\mathbf{F})
\big)
=
\big(
X_{i,\boldsymbol{\alpha}} ^{(\ell)}:
~\ell\in[L]
\big),
\label{ABO-encoding-function}
\end{align}
where
\begin{align}
&X_{i,\boldsymbol{\alpha}} ^{(\ell)}
\triangleq
\label{ABO-transmitted-symbol}\\\nonumber
&\begin{cases}
\alpha_i F_{i,\ell}
&\text{if}~|\boldsymbol{\alpha}|\leq N/2,
\\
\alpha_iF_{i,\ell}+\!\!
\displaystyle
\sum\nolimits_{n\in[N]\setminus\{i\}}
F_{n,\ell}
&\text{if}~
|\boldsymbol{\alpha}|>N/2~\text{and}~
\ell\in\mathcal{L}_i,
\\
(1-\alpha_i)F_{i,\ell},
&\text{if}~
|\boldsymbol{\alpha}|>N/2~\text{and}~
\ell\notin\mathcal{L}_i;
\end{cases}
\end{align}

\item
the decoding function
$\psi_{\boldsymbol{\alpha}}$ for each
$\boldsymbol{\alpha}\in\mathbb{F}_2^N$ at the receiver $\rho$
is given by
\begin{align*}
& \!\psi_{\boldsymbol{\alpha}}
(
\phi_{i,\boldsymbol{\alpha}}
(
\boldsymbol{F}_i,
\varphi_i (\mathbf{F})
):
i\in[N]
) \! \! = \! \!
(
\sum\nolimits_{i=1}^N
X_{i,\boldsymbol{\alpha}} ^{(\ell)}:
 \ell\in[L]
).
\end{align*}

\end{itemize}

We next show that the code $\mbC$ is admissible, i.e., for   each
$\boldsymbol{\alpha}=(\alpha_1,\alpha_2,\cdots,\alpha_N)
\in\mathbb{F}_2^N$ and   each $\ell\in[L]$,
$\sum_{i=1}^N X_{i,\boldsymbol{\alpha}} ^{(\ell)}=\sum_{i=1}^N \alpha_iF_{i,\ell};$
and
for any
$\boldsymbol{\alpha}=(\alpha_1,\alpha_2,\cdots,\alpha_N)
\in\mathbb{F}_2^N$,
\begin{align}
&\sum\nolimits_{i=1}^N
\log|\Im\,\phi_{i,\boldsymbol{\alpha}}|\label{ABO-code-R-alpha}
\\&=
\begin{cases}
|\boldsymbol{\alpha}|L
&\textup{if}~
|\boldsymbol{\alpha}|\leq N/2,
\\
\Big(
N-|\boldsymbol{\alpha}|
+
\frac{|\boldsymbol{\alpha}|}{N}
\Big)L
&\textup{if}~
|\boldsymbol{\alpha}|>N/2.
\end{cases}\nonumber
\end{align}
Towards this end, we consider two cases below.

\smallskip
\noindent\textbf{Case 1:}
$|\boldsymbol{\alpha}|\leq N/2$.

For this case, by \eqref{ABO-transmitted-symbol}, we have
$
\sum_{i=1}^N
X_{i,\boldsymbol{\alpha}}^{(\ell)}
=
\sum_{i=1}^N
\alpha_iF_{i,\ell} .
$ Moreover, by \eqref{ABO-encoding-function} and
\eqref{ABO-transmitted-symbol}, we can readily see that
\[
\log|\Im\,\phi_{i,\boldsymbol{\alpha}}| = \begin{cases} L &\quad\textup{if}~ \alpha_i\neq 0,\\ 0 &\quad\textup{if}~ \alpha_i=0. \end{cases}
\]
This immediately implies that
$
\sum_{i=1}^N
\log|\Im\,\phi_{i,\boldsymbol{\alpha}}|
=
|\boldsymbol{\alpha}|L.$

\smallskip
\noindent\textbf{Case 2:}
$|\boldsymbol{\alpha}|>N/2$.

Let $\textup{I}(\ell)\in[N]$ denote the unique
index such that
$
\ell\in\mathcal{L}_{\textup{I}(\ell)}.
$
By \eqref{ABO-transmitted-symbol}, we have
\begin{align*}
&\sum\nolimits_{i=1}^N
X_{i,\boldsymbol{\alpha}}^{(\ell)}
=
X_{\textup{I}(\ell),\boldsymbol{\alpha}}^{(\ell)}
+
\sum_{i\in[N]\setminus\{\textup{I}(\ell)\}}
X_{i,\boldsymbol{\alpha}}^{(\ell)}
\\
&=
\alpha_{\textup{I}(\ell)}
F_{\textup{I}(\ell),\ell}
+
\sum_{n\in[N]\setminus\{\textup{I}(\ell)\}}
F_{n,\ell} +
\sum_{i\in[N]\setminus\{\textup{I}(\ell)\}}
(1-\alpha_i)F_{i,\ell}
\\
&=
\alpha_{\textup{I}(\ell)}
F_{\textup{I}(\ell),\ell}
+
\sum_{i\in[N]\setminus\{\textup{I}(\ell)\}}
\alpha_iF_{i,\ell}
\\
&=
\sum\nolimits_{i=1}^N
\alpha_iF_{i,\ell},
\end{align*}
where the third equality follows from
$
1+(1-\alpha_i)=\alpha_i
$
over $\mathbb{F}_2$ for every $i\in[N]\setminus\{\textup{I}(\ell)\}$.
By
\eqref{eq:full-coordinate-partition}, \eqref{ABO-encoding-function} and
\eqref{ABO-transmitted-symbol}, we have
\[
\log|\Im\,\phi_{i,\boldsymbol{\alpha}}| = \begin{cases} L/N &\quad\textup{if}~\alpha_i\neq 0,\\ L &\quad\textup{if}~\alpha_i=0. \end{cases}
\]
Therefore, we have
\begin{align*}
\sum\nolimits_{i=1}^N
\log|\Im\,\phi_{i,\boldsymbol{\alpha}}|
&=
|\boldsymbol{\alpha}|\frac{L}{N}
+
\big(
N-|\boldsymbol{\alpha}|
\big)L
\nonumber\\
&=
\left(
N-|\boldsymbol{\alpha}|
+
\frac{|\boldsymbol{\alpha}|}{N}
\right)L.
\end{align*}

\subsection{Parity-Aided-Placement Code}
\label{subsec:PAP-construction}

In this subsection,
we consider the model $(\mathbb{F}_2,N,M)$ with $N\geq2$ and construct a parity-aided-placement (PAP) code  for this model by concatenating an uncached code, a full-placement code, and an ABO parity code.

We first present some notation.
Let $L_0,L_1,L_2$ be arbitrary nonnegative integers such that
$N$ divides both $L_1$ and $L_2$ and
\begin{align}
(N-1)L_1+L_2
\leq
NM(L_0+L_1+L_2).
\label{PAP-cache-constraint}
\end{align}
Let
$
L\triangleq L_0+L_1+L_2>0.
$
As will become clear below, the condition in
\eqref{PAP-cache-constraint} guarantees that the cache capacity
constraint is satisfied. Further, we let
\begin{align}
\lambda_0\triangleq\frac{L_0}{L},\quad\lambda_1\triangleq\frac{L_1}{L},\quad\text{ and }\quad
\lambda_2\triangleq\frac{L_2}{L}.\label{def-lambda}
\end{align}
Correspondingly, the inequality \eqref{PAP-cache-constraint} can be rewritten as
\begin{align}
(N-1)\lambda_1+\lambda_2\leq NM.\label{lambda-cond}
\end{align}

For every $i\in[N]$, partition the file as $\boldsymbol F_i=[\boldsymbol F_i^{(0)}\ \boldsymbol F_i^{(1)}\ \boldsymbol F_i^{(2)}]$, where $\boldsymbol F_i^{(j)}\in\mathbb F_2^{L_j}$ for $j=0,1,2$. Let $\mathbf F^{(j)}\in\mathbb F_2^{N\times L_j}$ be the matrix whose $i$th row is $\boldsymbol F_i^{(j)}$.

Furthermore, we let $\mbC^{(0)}$, $\mbC^{(1)}$, and $\mbC^{(2)}$ denote an $L_0$-dimensional uncached code for the model $(\mathbb{F}_2,N,0)$, an $L_1$-dimensional full-placement code for the model $\big(\mathbb{F}_2,N,(N-1)/N\big)$, and an $L_2$-dimensional ABO parity code for the model $(\mathbb{F}_2,N,1/N)$, respectively. Specifically, we write
\[
\mbC^{(j)}
=
\big\{
\varphi_i^{(j)},
\phi_{i,\boldsymbol{\alpha}}^{(j)},
\psi_{\boldsymbol{\alpha}}^{(j)}:
i\in[N],~
\boldsymbol{\alpha}\in\mathbb{F}_2^N
\big\},\quad j=0,1,2.
\]
 Based on these three codes, we construct an $L$-dimensional  PAP  code
$
\mbC
=
\big\{
\varphi_i,
\phi_{i,\boldsymbol{\alpha}},
\psi_{\boldsymbol{\alpha}}:
i\in[N],~
\boldsymbol{\alpha}\in\mathbb{F}_2^N
\big\}
$
for the model $(\mathbb{F}_2,N,M)$ as follows:
\begin{itemize}

\item
the caching function $\varphi_i$ at every source node $\sigma_i$ is
given by
$\varphi_i(\mathbf{F}) = \Big[ \varphi_i^{(j)}(\mathbf{F}^{(j)}):~j=0,1,2 \Big];$

\item
the encoding function $\phi_{i,\boldsymbol{\alpha}}$ for each
coefficient vector
$\boldsymbol{\alpha}\in\mathbb{F}_2^N$ at every source node
$\sigma_i$ is given by
\begin{align*}
&\phi_{i,\boldsymbol{\alpha}}
\big(
\boldsymbol{F}_i,\varphi_i(\mathbf{F})
\big)\!\! =\!\!
\Big[
\phi_{i,\boldsymbol{\alpha}}^{(j)}
\big(
\boldsymbol{F}_i^{(j)},
\varphi_i^{(j)}(\mathbf{F}^{(j)})
\big):~j=0,1,2
\Big];
\end{align*}

\item
the decoding function $\psi_{\boldsymbol{\alpha}}$ for each
$\boldsymbol{\alpha}\in\mathbb{F}_2^N$ at the receiver $\rho$ is
given by
\begin{align*}
&\psi_{\boldsymbol{\alpha}}
(
\phi_{i,\boldsymbol{\alpha}}
(
\boldsymbol{F}_i,\varphi_i(\mathbf{F})
):
i\in[N]
)
\nonumber\\
&=
\big[\psi_{\boldsymbol{\alpha}}^{(j)}
(
\phi_{i,\boldsymbol{\alpha}}^{(j)}
(
\boldsymbol{F}_i^{(j)},\varphi_i^{(j)}(\mathbf{F}^{(j)})
):
i\in[N]
):\,j=0,1,2
\big].
\end{align*}

\end{itemize}

Since the three elementary codes are admissible, we can readily see that
the code $\mbC$ is also admissible. For every $i\in[N]$, the caching functions act on disjoint file blocks, so
$|\Im\,\varphi_i|=\prod_{j=0}^2|\Im\,\varphi_i^{(j)}| =2^{[(N-1)L_1+L_2]/N}\leq2^{ML},$
where the inequality follows from~\eqref{PAP-cache-constraint}.

In the following,
we calculate the communication rate of the code $\mbC$.
Consider an arbitrary coefficient vector
$\boldsymbol{\alpha}\in\mathbb{F}_2^N$.
For each $i\in[N]$, we readily see that
$|\Im\,\phi_{i,\boldsymbol{\alpha}}|=|\Im\,\phi_{i,\boldsymbol{\alpha}}^{(0)}|\cdot |\Im\,\phi_{i,\boldsymbol{\alpha}}^{(1)}| \cdot |\Im\,\phi_{i,\boldsymbol{\alpha}}^{(2)}|.$
It follows that
\begin{align}
&\sum\nolimits_{i=1}^N
\log|\Im\,\phi_{i,\boldsymbol{\alpha}}|=
\sum\nolimits_{j=0}^2\sum\nolimits_{i=1}^N
\log|\Im\,\phi_{i,\boldsymbol{\alpha}}^{(j)}|.
\label{PAP-message-decomposition}
\end{align}

Substituting the elementary-code rates in~\eqref{un-code-R-alpha}, \eqref{full-code-R-alpha}, and~\eqref{ABO-code-R-alpha} into~\eqref{PAP-message-decomposition}, and using~\eqref{def-lambda}, gives
\begin{align}
&R_{\boldsymbol{\alpha}}(\mbC)
=
\frac{1}{L}
\sum\nolimits_{i=1}^N
\log|\Im\,\phi_{i,\boldsymbol{\alpha}}|
\label{PAP-R-alpha}\\
&=
\begin{cases}
0
&\textup{if}~
|\boldsymbol{\alpha}|=0,
\\ \lambda_0|\boldsymbol{\alpha}|+\lambda_1+\lambda_2|\boldsymbol{\alpha}|
&\textup{if}~
1\leq |\boldsymbol{\alpha}|\leq N/2,
\\
\lambda_0|\boldsymbol{\alpha}|
+ \lambda_1+
 \lambda_2 \big(
N-|\boldsymbol{\alpha}|
 +
\frac{|\boldsymbol{\alpha}|}{N}
\big)
&\textup{if}~
 |\boldsymbol{\alpha}| >N/2.
\end{cases}
\nonumber
\end{align}
We readily see that
$R_{\boldsymbol{\alpha}}(\mbC)$ depends only on
the Hamming weight $|\boldsymbol{\alpha}|$. 
For notational simplicity, we write $w\triangleq |\boldsymbol{\alpha}| $. Further,
we note that 
\begin{align}
g_{N}(w)&\triangleq\min\nolimits\left\{w, N-w +\frac{w}{N}\right\}\label{g-N}\\
&=
\begin{cases}
w  & \textup{if}~0\leq w\leq N/2,\\
N-w+\frac{w}{N} &\textup{if}~w>N/2.
\end{cases} \nonumber
\end{align}
With this, we can rewrite $R_{\boldsymbol{\alpha}}(\mbC)$ as
\[R_{\boldsymbol{\alpha}}(\mbC)= \begin{cases} 0 &\textup{if}~ w=0, \\ \lambda_0 w+\lambda_1+\lambda_2 g_N(w) &\textup{if}~ 1\leq w\leq N. \end{cases}\]
Consequently, we have 
\begin{align}
 R(\mbC)&=
\max\nolimits_{\boldsymbol{\alpha}\in\mathbb{F}_{2}^N}
R_{\boldsymbol{\alpha}}(\mbC)\nonumber\\
&=\max_{w\in[N]}
\big\{
\lambda_0w+\lambda_1
+\lambda_2 g_N(w)
\big\}.
\label{PAP-code-rate}
\end{align}  

We note that $L_0,L_1,L_2$ can be arbitrarily chosen
such that $N$ divides both $L_1$ and $L_2$
and~\eqref{PAP-cache-constraint} holds.
Thus, by~\eqref{def-lambda}, any nonnegative
$(\lambda_0,\lambda_1,\lambda_2)$ satisfying
$\lambda_0+\lambda_1+\lambda_2=1$
and~\eqref{lambda-cond}
can be approximated arbitrarily closely.
Together with~\eqref{PAP-code-rate},
this immediately gives the following theorem.

\begin{theorem}
Consider the model $(\mathbb{F}_2,N,M)$ with $N\geq2$. Let
$R_{\mathrm{PAP}}(N,M) \triangleq \min_{\lambda_0,\lambda_1,\lambda_2}~ \max_{w\in[N]} \big\{ \lambda_0w+\lambda_1+\lambda_2g_N(w) \big\},$
where the minimum is subject to
$\lambda_0,\lambda_1,\lambda_2\geq0, \lambda_0+\lambda_1+\lambda_2=1, (N-1)\lambda_1+\lambda_2\leq NM.$
Then,
$R (\mathbb{F}_2,N,M)\leq R_{\mathrm{PAP}}(N,M).$
\end{theorem}

 The following theorem gives  an explicit expression for
$R_{\mathrm{PAP}}(N,M)$,
whose proof is given in Appendix~\ref{app:pap-proof}.

% \begin{table}[!t]
% \caption{PAP placement and delivery for $N=3$, $M=1/3$, and $\textcolor{red}{\mathcal S_{\alpha}}=\{1,2\}$.}
% \label{tab:pap-example-placement}
% \centering
% \footnotesize
% \setlength{\tabcolsep}{3pt}
% \renewcommand{\arraystretch}{1.08}
% \begin{tabular}{@{}lccl@{}}
% \toprule
% Type (count) & Holder & Cached & Transmissions \\
% \midrule
% Uncached (6) & -- & -- & $\sigma_1:x_1,\ \sigma_2:x_2$ \\
% ABO (3) & $\sigma_1$ & $x_2\oplus x_3$ & $\sigma_1:t,\ \sigma_3:x_3$ \\
% ABO (3) & $\sigma_2$ & $x_1\oplus x_3$ & $\sigma_2:t,\ \sigma_3:x_3$ \\
% ABO (3) & $\sigma_3$ & $x_1\oplus x_2$ & $\sigma_3:x_1\oplus x_2$ \\
% Full (6) & $\sigma_h$ & $(x_i:i\ne h)$ & $\sigma_h:x_1\oplus x_2$ \\
% \bottomrule
% \end{tabular}
% \par\smallskip
% \parbox{\columnwidth}{\footnotesize At coordinate $\ell$, write $x_i=F_{i,\ell}$ and $t=x_1\oplus x_2\oplus x_3$. Counts specify numbers of coordinates. Each source holds two full and three ABO coordinates, caching $2\cdot2+3=7$ bits. Only the listed sources transmit.}
% \end{table}

\begin{theorem} 
For the model $(\mathbb{F}_2,N,M)$ with $N\geq3$,
$R_{\mathrm{PAP}}(N,M)$ is given as follows.

\noindent$\bullet$ If $N$ is odd,   $R_{\mathrm{PAP}}(N,M)$ is
\begin{equation*}
\begin{cases}
N-N(N-1)M 
&\textup{if}~M\in[0,\Msmall],\\
1+\frac{(N-1)^2(N-1-NM)}{2N^2-4N+1} 
&\textup{if}~M\in[\Msmall,\Mfull],\\
1 &\textup{if}~M\geq\Mfull.
\end{cases}
%\label{eq:pap-odd-rate}
\end{equation*}

\noindent$\bullet$ 
If $N$ is even,  $R_{\mathrm{PAP}}(N,M)$ is
\begin{equation*}
\begin{cases}
N-N(N-1)M 
&\textup{if}~M\in[0,\Msmall],\\
\frac{N}{2}+1-\frac{3N-2}{2}M 
&\textup{if}~M\in[\Msmall,\Meven],\\
\frac{N(3N+1)-(3N-2)NM}{2(3N-1)} 
&\textup{if}~M\in[\Meven,\Mfull],\\
1 &\textup{if}~M\geq\Mfull.
\end{cases}
%\label{eq:pap-even-rate}
\end{equation*}
\label{thm:pap-minimax}
\end{theorem}

Combining this achievable rate with the converse gives the following exact characterization.

\begin{corollary}
For every $N\geq3$,
\begin{equation*}
R(\Ftwo,N,M)=
\begin{cases}
N-N(N-1)M &\textup{if}~0\leq M\leq\frac{1}{2N-1},\\
1 &\textup{if}~M\geq\frac{N-1}{N}.
\end{cases}
%\label{eq:exact-regimes}
\end{equation*}
\label{cor:exact-regimes}
\end{corollary}

\begin{figure}[!t]
\centering
\includegraphics[width=0.94\columnwidth]{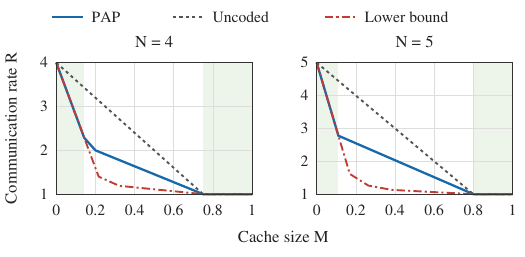}
\caption{Binary memory--rate tradeoffs for $N=4,5$. Shaded intervals mark the exact regimes in Corollary~\ref{cor:exact-regimes}. The uncoded baseline uses only full and uncached coordinates.}
\label{fig:pap-comparison}
\end{figure}

\section{Numerical Results and Discussion}
\label{sec:numerical}

Figure~\ref{fig:pap-comparison} compares the PAP rate in Theorem~\ref{thm:pap-minimax} with the converse in Theorem~\ref{thm:multi-query-converse} for $N=4,5$. The maximum average ranks are evaluated by enumerating binary query row spaces for each pair $(m,n)$. We also consider an uncoded baseline using only full and uncached coordinates, which achieves the worst-case rate $\max\{1,N(1-M)\}$.

For $0\leq M\leq1/(2N-1)$, PAP meets the converse and reduces the rate with increasing cache size at $N-1$ times the slope magnitude of the baseline. In this regime, full-support requests determine the worst-case rate. An ABO coordinate serves such a request with one transmitted bit using one cached parity and the holder's local symbol, whereas a full coordinate requires caching $N-1$ symbols to achieve the same delivery cost. For $M\geq(N-1)/N$, both schemes achieve the optimal rate of one.

For $N=4,5$, PAP improves upon the baseline at intermediate cache sizes but does not meet the lower bound. At $N=5$ and $M=0.5$, for example, the PAP rate is $1.774$, compared with $2.5$ for the baseline and a lower bound of $1.094$. The optimal memory--rate tradeoff in this intermediate regime remains open.

\section{Conclusion}

We formulated a cache-aided linear computation model in which each source retains its own file and caches additional information before the requested linear combination is known. We established a converse bound over arbitrary finite fields by considering multiple requests under the same cache placement. For binary files, we proposed parity-aided placement (PAP) and obtained a closed-form worst-case rate for every $N\geq3$ by optimizing its placement parameters. The resulting achievable rate meets the converse for $0\leq M\leq1/(2N-1)$ and $M\geq(N-1)/N$, establishing the optimal memory--rate tradeoff in these two regimes.
\clearpage
\section*{Acknowledgment}
The authors used ChatGPT to assist with language editing throughout the manuscript and to provide supplementary checks of mathematical arguments, particularly the proof of Theorem~\ref{thm:pap-minimax}. The research conception, model formulation, code constructions, proof development, and initial manuscript were the original work of the authors. The authors reviewed and revised the AI-assisted edits, independently verified the mathematical arguments and results, and take full responsibility for the final content.

\bibliographystyle{IEEEtran}
\bibliography{references}

@article{MaddahAliNiesen2014,
  author  = {M. A. Maddah-Ali and U. Niesen},
  title   = {Fundamental Limits of Caching},
  journal = {IEEE Transactions on Information Theory},
  volume  = {60},
  number  = {5},
  pages   = {2856--2867},
  month   = may,
  year    = {2014},
  doi     = {10.1109/TIT.2014.2306938}
}

@article{YuMaddahAliAvestimehr2018,
  author  = {Q. Yu and M. A. Maddah-Ali and A. S. Avestimehr},
  title   = {The Exact Rate-Memory Tradeoff for Caching With Uncoded Prefetching},
  journal = {IEEE Transactions on Information Theory},
  volume  = {64},
  number  = {2},
  pages   = {1281--1296},
  month   = feb,
  year    = {2018},
  doi     = {10.1109/TIT.2017.2785237}
}

@article{WanSunJiTuninettiCaire2021,
  author  = {K. Wan and H. Sun and M. Ji and D. Tuninetti and G. Caire},
  title   = {On the Optimal Load-Memory Tradeoff of Cache-Aided Scalar Linear Function Retrieval},
  journal = {IEEE Transactions on Information Theory},
  volume  = {67},
  number  = {6},
  pages   = {4001--4018},
  month   = jun,
  year    = {2021},
  doi     = {10.1109/TIT.2021.3066005}
}

@article{WanSunJiTuninettiCaire2021GLFR,
  author  = {K. Wan and H. Sun and M. Ji and D. Tuninetti and G. Caire},
  title   = {Cache-Aided General Linear Function Retrieval},
  journal = {Entropy},
  volume  = {23},
  number  = {1},
  pages   = {25},
  year    = {2021},
  doi     = {10.3390/e23010025}
}

@article{JiCaireMolisch2016,
  author  = {M. Ji and G. Caire and A. F. Molisch},
  title   = {Fundamental Limits of Caching in Wireless {D2D} Networks},
  journal = {IEEE Transactions on Information Theory},
  volume  = {62},
  number  = {2},
  pages   = {849--869},
  month   = feb,
  year    = {2016},
  doi     = {10.1109/TIT.2015.2504556}
}

@article{YuMaddahAliAvestimehr2019,
  author  = {Q. Yu and M. A. Maddah-Ali and A. S. Avestimehr},
  title   = {Characterizing the Rate-Memory Tradeoff in Cache Networks Within a Factor of 2},
  journal = {IEEE Transactions on Information Theory},
  volume  = {65},
  number  = {1},
  pages   = {647--663},
  month   = jan,
  year    = {2019},
  doi     = {10.1109/TIT.2018.2870566}
}

@article{YaparWanSchaeferCaire2019,
  author  = {{\c{C}}. Yapar and K. Wan and R. F. Schaefer and G. Caire},
  title   = {On the Optimality of {D2D} Coded Caching With Uncoded Cache Placement and One-Shot Delivery},
  journal = {IEEE Transactions on Communications},
  volume  = {67},
  number  = {12},
  pages   = {8179--8192},
  month   = dec,
  year    = {2019},
  doi     = {10.1109/TCOMM.2019.2940671}
}

@article{LiMaddahAliYuAvestimehr2018,
  author  = {S. Li and M. A. Maddah-Ali and Q. Yu and A. S. Avestimehr},
  title   = {A Fundamental Tradeoff Between Computation and Communication in Distributed Computing},
  journal = {IEEE Transactions on Information Theory},
  volume  = {64},
  number  = {1},
  pages   = {109--128},
  month   = jan,
  year    = {2018},
  doi     = {10.1109/TIT.2017.2756959}
}

@article{WanSunJiCaire2022,
  author  = {K. Wan and H. Sun and M. Ji and G. Caire},
  title   = {Distributed Linearly Separable Computation},
  journal = {IEEE Transactions on Information Theory},
  volume  = {68},
  number  = {2},
  pages   = {1259--1278},
  month   = feb,
  year    = {2022},
  doi     = {10.1109/TIT.2021.3127910}
}

@article{AppuswamyFranceschettiKaramchandaniZeger2011,
  author  = {R. Appuswamy and M. Franceschetti and N. Karamchandani and K. Zeger},
  title   = {Network Coding for Computing: Cut-Set Bounds},
  journal = {IEEE Transactions on Information Theory},
  volume  = {57},
  number  = {2},
  pages   = {1015--1030},
  month   = feb,
  year    = {2011},
  doi     = {10.1109/TIT.2010.2095070}
}

@article{AppuswamyFranceschetti2014,
  author  = {R. Appuswamy and M. Franceschetti},
  title   = {Computing Linear Functions by Linear Coding Over Networks},
  journal = {IEEE Transactions on Information Theory},
  volume  = {60},
  number  = {1},
  pages   = {422--431},
  month   = jan,
  year    = {2014},
  doi     = {10.1109/TIT.2013.2283075}
}

@article{TianChen2018,
  author  = {C. Tian and J. Chen},
  title   = {Caching and Delivery via Interference Elimination},
  journal = {IEEE Transactions on Information Theory},
  volume  = {64},
  number  = {3},
  pages   = {1548--1560},
  month   = mar,
  year    = {2018},
  doi     = {10.1109/TIT.2018.2794543}
}

@article{GuangYeungYangLi2019,
  author  = {X. Guang and R. W. Yeung and S. Yang and C. Li},
  title   = {Improved Upper Bound on the Network Function Computing Capacity},
  journal = {IEEE Transactions on Information Theory},
  volume  = {65},
  number  = {6},
  pages   = {3790--3811},
  month   = jun,
  year    = {2019},
  doi     = {10.1109/TIT.2019.2893107}
}

@article{GuangZhang2023,
  author  = {X. Guang and R. Zhang},
  title   = {Zero-Error Distributed Compression of Binary Arithmetic Sum},
  journal = {IEEE Transactions on Information Theory},
  volume  = {70},
  number  = {5},
  pages   = {3100--3117},
  month   = may,
  year    = {2024},
  doi     = {10.1109/TIT.2023.3319976}
}

@article{ValenteEramoLavacca2023,
  author  = {F. Valente and V. Eramo and F. G. Lavacca},
  title   = {Optimal bandwidth and computing resource allocation in low earth orbit satellite constellation for earth observation applications},
  journal = {Computer Networks},
  volume  = {232},
  pages   = {109849},
  year    = {2023},
  doi     = {10.1016/j.comnet.2023.109849}
}

@article{XingWangTengLi2022,
  author  = {N. Xing and Y. Wang and L. Teng and L. Li},
  title   = {Communication and trajectory design in {UAV}-enabled flying network},
  journal = {Digital Signal Processing},
  volume  = {126},
  pages   = {103489},
  year    = {2022},
  doi     = {10.1016/j.dsp.2022.103489}
}

\clearpage
\appendices

\section{Proof of the Converse}\label{app:converse-proof}
\begin{IEEEproof}
Consider an admissible $L$-dimensional code
$
\mbC=\big\{\varphi_i,\phi_{i,\boldsymbol{\alpha}},\psi_{\boldsymbol{\alpha}}:
i\in[N],~\boldsymbol{\alpha}\in\mathbb{F}_q^N\big\}
$
for the model $(\mathbb{F}_q,N,M)$.
We let $m,\,n$ be arbitrary positive integers with $1\leq m,\,n\leq N$. Further, we let $\boldsymbol{\alpha}_1,\boldsymbol{\alpha}_2,\cdots, \boldsymbol{\alpha}_{m}$
be $m$ arbitrary coefficient vectors.
We  consider the communication rate of the code~$\mbC$
\begin{align}
R(\mbC)&=\max_{\boldsymbol{\alpha} \in \mathbb{F}_q^N} R_{\boldsymbol{\alpha}}(\mbC) \nonumber\\
&\geq \frac{1}{m}\cdot \sum_{k=1}^{m} R_{\boldsymbol{\alpha}_k}(\mbC)\nonumber\\
&=\frac{1}{m}\cdot \sum_{k=1}^{m} \frac{1}{\,L\,}\cdot
\sum_{i=1}^N  \log|\Im\,\phi_{i,\boldsymbol{\alpha}_k}|\nonumber\\
&\geq
\frac{1}{mL}\cdot\sum_{k=1}^{m} \sum_{i=1}^N H\Big(\phi_{i,\boldsymbol{\alpha}_k}\big( \boldsymbol{F}_i ,\,\varphi_i(\mathbf{F})\big)\Big),\label{pf-thm-lower-bound-eq4}
\end{align}
where $H(\,\cdot\,)$ denotes the entropy function computed with logarithm base $q$.

For each $i\in [N]$ and $k\in[m]$, let
$X_{i,\boldsymbol{\alpha}_k}$
be the encoded message from the source node $\sigma_i$ to the receiver $\rho$ associated with the coefficient vector $\boldsymbol{\alpha}_k$, i.e.,
\begin{align}
X_{i,\boldsymbol{\alpha}_k}\triangleq \phi_{i,\boldsymbol{\alpha}_k}\big( \boldsymbol{F}_i ,\,\varphi_i(\mathbf{F})\big).\label{pf-thm-lower-bound-eq1}
\end{align}
With this, we further let
\begin{align}
\boldsymbol{X}_i
\triangleq
\big(X_{i,\boldsymbol{\alpha}_1},X_{i,\boldsymbol{\alpha}_2},\cdots,X_{i,\boldsymbol{\alpha}_{m}}\big),\qquad\forall~i\in [N].\label{boldsymbol-X_i}
\end{align}
Continuing from \eqref{pf-thm-lower-bound-eq4}, we have
\begin{align}
R(\mbC)
&\geq
\frac{1}{mL}\cdot
\sum_{i=1}^N \sum_{k=1}^{m} H(X_{i,\boldsymbol{\alpha}_k}) \geq
\frac{1}{mL}\cdot
\sum_{i=1}^N H(\boldsymbol{X}_i),
\label{eq:multi-query-proof-upper-final}
\end{align}
where the last inequality in \eqref{eq:multi-query-proof-upper-final} follows from the subadditivity of entropy.

Next, we lower bound $\sum_{i=1}^N H(\boldsymbol{X}_i).$
For notational simplicity,
we let
\begin{align*}
\boldsymbol{X}_{\theta}
\triangleq
(\boldsymbol{X}_i:\;i\in \theta),\qquad\forall~\theta\subseteq[N].
\end{align*}
In particular, we let $\mathbf{X}\triangleq
(\boldsymbol{X}_i:\;i\in [N])$.
Then, we have
\begin{align}
\sum_{i=1}^N H(\boldsymbol{X}_i)
&=\frac{1}{\binom{N-1}{n-1}}
\cdot
\sum_{\theta\subseteq [N]\;\text{s.t.}\;|\theta|=n}~ \sum_{i\in\theta} ~ H(\boldsymbol{X}_i)\label{pf-thm-lower-bound-eq9-0}\\
&\geq\frac{1}{\binom{N-1}{n-1}}
\cdot
\sum_{\theta\subseteq [N]\;\text{s.t.}\;|\theta|=n}~ H(\boldsymbol{X}_\theta),
\label{pf-thm-lower-bound-eq9}
\end{align}
where the equality \eqref{pf-thm-lower-bound-eq9-0} follows because
for any given $i\in[N]$, the number of subsets
$\theta\subseteq[N]$ satisfying $|\theta|=n$  and $i\in\theta$   is equal to $\binom{N-1}{n-1}$, and
the  inequality \eqref{pf-thm-lower-bound-eq9} follows from the subadditivity of entropy.

Furthermore,
we  consider $H(\boldsymbol{X}_{\theta})$ for a subset $\theta\subseteq[N]$ with $|\theta|=n$.
For the  file matrix $\mathbf{F}\in\mathbb{F}_q^{N\times L}$, we let
$\boldsymbol{F}_{\theta}$ be the submatrix of $\mathbf{F}$ containing the $i$th row if $i\in \theta$, and let $\boldsymbol{F}_{\theta^{\textup{c}}}$ be the submatrix of $\mathbf{F}$ containing the $i$th row if $i\in \theta^{\textup{c}}\triangleq [N]\setminus\theta$.
Then,
\begin{align}
 H(\boldsymbol{X}_\theta)
\geq
H(\boldsymbol{X}_\theta\mid \boldsymbol{F}_{\theta^{\textup{c}}}) =
H(\mathbf{X}\mid \boldsymbol{F}_{\theta^{\textup{c}}})
-
H(\boldsymbol{X}_{\theta^{\textup{c}}}
\mid
\boldsymbol{X}_\theta,\boldsymbol{F}_{\theta^{\textup{c}}}),
\label{eq:multi-query-proof-chain}
\end{align}
where   the
equality in \eqref{eq:multi-query-proof-chain} follows from the chain rule of entropy.

We first consider
$H(\mathbf{X}\mid \boldsymbol{F}_{\theta^{\textup{c}}})$
in \eqref{eq:multi-query-proof-chain}. Note that for each  coefficient vector
$\boldsymbol{\alpha}_k$ with $k\in[m]$, the code~$\mbC$ can compute the linear function
$ \boldsymbol{\alpha}_k\mathbf{F} $ with zero error, i.e.,
\begin{align}
&\psi_{\boldsymbol{\alpha}_k}\big(X_{i,\boldsymbol{\alpha}_k},\,i\in[N]\big)\nonumber\\
&=\psi_{\boldsymbol{\alpha}_k}\Big(\phi_{i,\boldsymbol{\alpha}_k}(\boldsymbol F_i,\varphi_i(\mathbf F)):i\in[N]\Big)\nonumber\\ &=\boldsymbol\alpha_k\mathbf F,\quad\forall k\in[m],\label{pf-thm-lower-bound-eq11}
\end{align}
where the first equality in \eqref{pf-thm-lower-bound-eq11} follows from \eqref{pf-thm-lower-bound-eq1}.
We let
\begin{align}
A\triangleq
\setlength{\arraycolsep}{3pt}
\renewcommand{\arraystretch}{0.7}
\begin{bmatrix}
\boldsymbol{\alpha}_1\\
\boldsymbol{\alpha}_2\\
\vdots\\
\boldsymbol{\alpha}_{m}
\end{bmatrix}
\in\mathbb{F}_q^{m\times N}.\label{pf-def-A}
\end{align}
By \eqref{pf-thm-lower-bound-eq11}, we can see that the matrix $A \mathbf{F}$ is a
function of
\begin{align*}
\mathbf{X}=(\boldsymbol{X}_1,\boldsymbol{X}_2,\cdots,\boldsymbol{X}_N)=\big( (X_{i,\boldsymbol{\alpha}_k}:\;k\in[m] ):\;i\in[N]\big).
\end{align*}
With this, we further have
\begin{align}
H (\mathbf{X}\, | \,\boldsymbol{F}_{\theta^{\textup{c}}} )
&\geq H(A \mathbf{F}\,| \,\boldsymbol{F}_{\theta^{\textup{c}}})
\nonumber\\&=H\big(A[\theta]\cdot\boldsymbol{F}_{\theta}+ A[\theta^{\textup{c}}]\cdot\boldsymbol{F}_{\theta^{\textup{c}}}\,\big| \,\boldsymbol{F}_{\theta^{\textup{c}}}\big)\nonumber\\
&=H\big(A[\theta]\cdot\boldsymbol{F}_{\theta}\big| \,\boldsymbol{F}_{\theta^{\textup{c}}}\big)\nonumber\\
&=
H(A[\theta]\cdot\boldsymbol{F}_{\theta})\label{pf-thm-lower-bound-eq12-00}\\
&=
L\cdot \mathrm{rank}(A[\theta]),\label{pf-thm-lower-bound-eq12-0}
\end{align}
where the two equalities \eqref{pf-thm-lower-bound-eq12-00} and \eqref{pf-thm-lower-bound-eq12-0} follow  because
the files $\boldsymbol{F}_1,\boldsymbol{F}_2,\cdots,\boldsymbol{F}_N$ are mutually independent and
uniformly distributed over $\mathbb{F}_q^L$.

On the other hand, we consider $H(\boldsymbol{X}_{\theta^{\textup{c}}}
\mid
\boldsymbol{X}_\theta,\boldsymbol{F}_{\theta^{\textup{c}}})$ in \eqref{eq:multi-query-proof-chain}. To be specific, it is easy to see that
\begin{align}
&H(\boldsymbol{X}_{\theta^{\textup{c}}}
\mid
\boldsymbol{X}_\theta,\boldsymbol{F}_{\theta^{\textup{c}}})\nonumber\\&\leq
H(\boldsymbol{X}_{\theta^{\textup{c}}}\mid\boldsymbol{F}_{\theta^{\textup{c}}})\nonumber\\&\leq
\sum_{i\in\theta^{\textup{c}}}
H(\boldsymbol X_i\mid\boldsymbol F_{\theta^{\textup c}})\nonumber\\&\leq
\sum_{i\in\theta^{\textup{c}}}
H(\boldsymbol{X}_i\mid \boldsymbol{F}_i).
\label{eq:multi-query-proof-second-term-1}
\end{align}
We further note that
for each $i\in\theta^c$, the random vector $\boldsymbol{X}_i=\big(X_{i,\boldsymbol{\alpha}_1},X_{i,\boldsymbol{\alpha}_2},\cdots,X_{i,\boldsymbol{\alpha}_{m}}\big)$ is a deterministic function of
$
\big(\boldsymbol{F}_i,\varphi_i(\mathbf{F})\big)
$ (cf.~\eqref{pf-thm-lower-bound-eq1}), which implies that $H(\boldsymbol{X}_i\mid \boldsymbol{F}_i)
 \leq
H\big(\varphi_i(\mathbf{F})\mid \boldsymbol{F}_i\big) $.
Continuing from~\eqref{eq:multi-query-proof-second-term-1}, we obtain that
\begin{align}
H(\boldsymbol{X}_{\theta^{\textup{c}}}
\mid
\boldsymbol{X}_\theta,\boldsymbol{F}_{\theta^{\textup{c}}})
&\leq
\sum_{i\in\theta^{\textup{c}}}
H(\boldsymbol{X}_i \; |\; \boldsymbol{F}_i) \nonumber\\
&\leq
\sum_{i\in\theta^{\textup{c}}}
H\big(\varphi_i(\mathbf{F}) \;\big|\; \boldsymbol{F}_i\big) \nonumber\\
&\leq
\sum_{i\in\theta^{\textup{c}}}
H\big(\varphi_i(\mathbf{F}) \big) \nonumber\\
&\leq \sum_{i\in\theta^{\textup{c}}} ML=(N-n)ML,\label{eq:multi-query-proof-cache}
\end{align}
where the inequality \eqref{eq:multi-query-proof-cache} follows because
$H\big(\varphi_i(\mathbf{F}) \big)\leq \log|\Im\,\varphi_i|\leq ML$ for each $i\in [N]$.

For each    subset $\theta\subseteq[N]$ with $|\theta|=n$,
by \eqref{eq:multi-query-proof-chain},
\eqref{pf-thm-lower-bound-eq12-0}, and
\eqref{eq:multi-query-proof-cache}, we obtain
\begin{align}
 H(\boldsymbol{X}_\theta)
\geq
L\cdot \big(\mathrm{rank} (A[\theta] )
-
(N-n)M\big).
\label{eq:multi-query-proof-subset-lower}
\end{align}
Combining \eqref{eq:multi-query-proof-subset-lower}
with \eqref{pf-thm-lower-bound-eq9}, we have
\begin{align}
&\sum_{i=1}^N H(\boldsymbol{X}_i)\nonumber\\
&\geq
\frac{1}{\binom{N-1}{n-1}}
\sum_{\theta\subseteq[N]:\,|\theta|=n}
L\cdot \big(
\mathrm{rank} (A[\theta] )
-
(N-n)M
\big) \nonumber\\
&=
\frac{L}{\binom{N-1}{n-1}}\cdot\binom{N}{n}\cdot
\left[\begin{aligned}
\frac{1}{\binom{N}{n}}
\sum_{\theta\subseteq[N]:\,|\theta|=n}
  \mathrm{rank} (A[\theta] )
\\{}-\frac{1}{\binom Nn}\sum_{\theta\subseteq[N]:|\theta|=n}
(N-n)M
\end{aligned}\right]\nonumber\\
&=
\frac{LN}{n}\cdot
\Big[ \avgrank_{n}(A)
-
(N-n)M
\Big],
\label{eq:multi-query-proof-lower-final}
\end{align}
where the equality
\eqref{eq:multi-query-proof-lower-final}
follows from $\binom{N}{n}/\binom{N-1}{n-1}=N/n$, the definition of
$\avgrank_{n}(A)$ (cf.~\eqref{eq:average-rank}), and
\[\begin{aligned}&\frac{1}{\binom Nn}\sum_{\theta\subseteq[N]:|\theta|=n}
(N-n)M\\&=\frac{1}{\binom Nn}\binom Nn(N-n)M
\\&=(N-n)M.
\end{aligned}\]

Continuing from  \eqref{eq:multi-query-proof-upper-final}, by
\eqref{eq:multi-query-proof-lower-final} we have
\begin{align}
R(\mbC)
&\geq
\frac{N}{mn}\cdot
\big( \avgrank_{n}(A)
-
(N-n)M
\big).
\label{eq:multi-query-proof-code-lb-A}
\end{align}
We note that the inequality \eqref{eq:multi-query-proof-code-lb-A} holds for an arbitrary matrix $A\in\mathbb{F}_q^{m\times N}$, implying that
\begin{align}
R(\mbC)
&\geq \max_{\text{all matrices $A\in\mathbb{F}_q^{m\times N}$}}\;\frac{\;N\;}{mn}\cdot
\big( \avgrank_{n}(A)
-
(N-n)M
\big)
\nonumber\\
&=
\frac{\;N\;}{mn}\cdot\Big(  \avgrank_{m,n,N}
-  (N-n)M\Big).
\label{eq:multi-query-proof-code-lb}
\end{align}
Furthermore, the inequality \eqref{eq:multi-query-proof-code-lb} is true for arbitrary positive integers $1\leq m,\,n\leq N$. As such, we have
\begin{align}
R(\mbC)
\ge\max_{1\leq m,\,n\leq N}\;
\frac{\;N\;}{mn}\cdot\Big(  \avgrank_{m,n,N}
-  (N-n)M\Big).\label{eq-pof}
\end{align}

We finally note that   the above lower bound \eqref{eq-pof} is valid for any positive integer $L$
and  any  admissible $L$-dimensional code
$
\mbC $
for the model $(\mathbb{F}_q,N,M)$.
The theorem is proved.
\end{IEEEproof}
\section{Detailed Derivation of the PAP Minimax Solution}\label{app:pap-proof}
For $N=2$, ABO placement coincides with full placement because the cached parity contains only the other source's symbol. Combining uncached and full-placement coordinates with full-placement fraction $\lambda_1=\min\{2M,1\}$ achieves worst-case rate $2-\lambda_1$. The $m=n=1$ and $m=n=2$ instances of~\eqref{eq:main-multi-query} give the matching lower bound, so $R(\Ftwo,2,M)=R_{\mathrm{PAP}}(2,M)=\max\{2-2M,1\}$ for all $M\geq0$.

\begin{proof}[Proof of Theorem~\ref{thm:pap-minimax}]
For a nonzero request of weight $w$, write
\begin{equation*}
r_{\mathrm{PAP}}(w;\lambda_1,\lambda_2)
\triangleq (1-\lambda_1-\lambda_2)w+\lambda_1+\lambda_2g_N(w).
\end{equation*}
This is the per-request rate in~\eqref{PAP-R-alpha}, with $\lambda_0=1-\lambda_1-\lambda_2$.

For $M\ge (N-1)/N$, choosing $(\lambda_1,\lambda_2)=(1,0)$ makes every coordinate a
full coordinate and achieves rate $1$. Moreover,
$r_{\mathrm{PAP}}(1;\lambda_1,\lambda_2)=1$ for every feasible pair
$(\lambda_1,\lambda_2)$, so no PAP code can have worst-case rate below $1$.

It remains to consider $0\leq M\leq(N-1)/N$. Over this range, we first reduce the optimization to one parameter. The three coordinate types are ordered by their per-source cache costs: an uncached coordinate requires no cache, a parity coordinate requires $1/N$, and a full coordinate requires $(N-1)/N$. Moving upward in this order never increases any per-request cost: replacing an uncached coordinate by a parity coordinate changes its cost from $w$ to $g_N(w)\leq w$, while replacing a parity coordinate by a full coordinate changes its cost from $g_N(w)\geq1$ to $1$.

Consequently, any feasible pair with $\frac{(N-1)\lambda_1+\lambda_2}{N}<M$ can be modified to increase its cache usage without increasing the value of the objective function, until the cache constraint is met with equality. It is therefore sufficient to optimize over pairs satisfying $\frac{(N-1)\lambda_1+\lambda_2}{N}=M$. By \eqref{lambda-cond},
\begin{equation}
\label{eq:PAP-active-memory}
(N-1)\lambda_1+\lambda_2=NM.
\end{equation}

We eliminate $\lambda_2=NM-(N-1)\lambda_1$ and write
\[
r_{\mathrm{PAP}}(w;\lambda_1)
\triangleq
r_{\mathrm{PAP}}\bigl(w;\lambda_1,NM-(N-1)\lambda_1\bigr).
\]
After this substitution, the nonnegativity constraints $\lambda_1\geq0$, $\lambda_2\geq0$, and $1-\lambda_1-\lambda_2\geq0$ reduce to the feasible interval
\begin{equation}
\label{eq:PAP-theta-feasible-interval}
\max\left\{0,\frac{NM-1}{N-2}\right\}
\le \lambda_1 \le \frac{NM}{N-1}.
\end{equation}

The resulting one-dimensional optimization is
\[
\min_{\lambda_1}\ \max_{w\in[N]}r_{\mathrm{PAP}}(w;\lambda_1),
\]
where $\lambda_1$ ranges over the interval above. It remains to characterize and minimize the upper envelope of the functions $r_{\mathrm{PAP}}(w;\lambda_1)$.

We next determine which request sizes $w\in[N]$ can attain $\max_{u\in[N]}r_{\mathrm{PAP}}(u;\lambda_1)$. For a fixed $w$, substituting $\lambda_2=NM-(N-1)\lambda_1$ into \eqref{PAP-R-alpha} gives
\[
\begin{aligned}
r_{\mathrm{PAP}}(w;\lambda_1)
&=\lambda_1+[NM-(N-1)\lambda_1]g_N(w)
\\&\quad+[1-NM+(N-2)\lambda_1]w\\
&=\bigl[w+NM(g_N(w)-w)\bigr]
\\&\quad+\lambda_1\bigl[1-(N-1)g_N(w)+(N-2)w\bigr].
\end{aligned}
\]
Thus, each request size defines an affine function of $\lambda_1$, and the worst-case rate is their upper envelope. To determine where the maximum over $w$, i.e., $\max_{w\in[N]}r_{\mathrm{PAP}}(w;\lambda_1)$, is attained, we separately consider $w\leq N/2$, for which $g_N(w)=w$, and $w>N/2$, for which $g_N(w)=N-w+w/N$.

\par\vspace{1mm}
\noindent\textbf{\textit{Requests with $w\leq N/2$.}}\par
In this regime,
\begin{equation}
\label{eq:PAP-low-support-affine-theta}
\begin{aligned}r_{\mathrm{PAP}}(w;\lambda_1)
&=\lambda_1+[NM-(N-1)\lambda_1]w\\&\quad+[1-NM+(N-2)\lambda_1]w
\\&=w+\lambda_1(1-w)
\\&=\lambda_1+(1-\lambda_1)w.\end{aligned}
\end{equation}
For $w\geq2$, this rate is strictly decreasing in $\lambda_1$. Moreover, $\lambda_1\leq1$ throughout the feasible interval, so the last expression is non-decreasing in $w$. Hence, the maximum over $w\leq N/2$ is attained at
\[
w_L
\triangleq
\left\lfloor\frac N2\right\rfloor.
\]

\par\vspace{1mm}
\noindent\textbf{\textit{Requests with $w>N/2$.}}\par
In this regime, a direct computation gives
\[
\begin{aligned}r_{\mathrm{PAP}}(w;\lambda_1)
&=w+M[N^2-w(2N-1)]\\
&\quad+\lambda_1\Bigl[1-N^2+N+\tfrac{w(2N^2-4N+1)}N\Bigr].\end{aligned}
\]
To identify the maximum over $w>N/2$, collect the terms that are independent of $w$ and proportional to $w$, respectively:
\begin{equation}
\label{eq:PAP-high-support-affine-s}
\begin{aligned}
r_{\mathrm{PAP}}(w;\lambda_1)
&=N^2M+\lambda_1(1-N^2+N)\\
&\quad
+w\left[
1-(2N-1)M
+\frac{2N^2-4N+1}{N}\lambda_1
\right].
\end{aligned}
\end{equation}
Define
\begin{equation}
\label{eq:PAP-AG-definitions}
\begin{aligned}
A(\lambda_1)
&\triangleq
N^2M+\lambda_1(1-N^2+N),\\
G(\lambda_1)
&\triangleq
1-(2N-1)M+\frac{2N^2-4N+1}{N}\lambda_1.
\end{aligned}
\end{equation}
Then \eqref{eq:PAP-high-support-affine-s} becomes
\[
r_{\mathrm{PAP}}(w;\lambda_1)
=
A(\lambda_1)+wG(\lambda_1).
\]
For $w>N/2$, the smallest admissible integer request size is
\[
w_H
\triangleq
\left\lfloor\frac N2\right\rfloor+1.
\]
By \eqref{eq:PAP-high-support-affine-s}, $r_{\mathrm{PAP}}(w;\lambda_1)$ is affine in $w$ over the integers from $w_H$ to $N$. Hence,
\[
\max_{w>N/2}r_{\mathrm{PAP}}(w;\lambda_1)
=
\max\left\{
r_{\mathrm{PAP}}(w_H;\lambda_1),
r_{\mathrm{PAP}}(N;\lambda_1)
\right\}.
\]
\par\vspace{1mm}
Combining this identity with the maximum over $w\leq N/2$ gives
\begin{equation}
\label{eq:PAP-three-candidate-envelope}
\begin{aligned}\max_{w\in[N]}r_{\mathrm{PAP}}(w;\lambda_1)
&=
\max\bigl\{
r_{\mathrm{PAP}}(w_L;\lambda_1),\\&\quad
r_{\mathrm{PAP}}(w_H;\lambda_1),
r_{\mathrm{PAP}}(N;\lambda_1)
\bigr\}.\end{aligned}
\end{equation}

\par\vspace{1mm}
Having reduced the worst-case rate to three affine functions, we now minimize their upper envelope by comparing how these functions vary with $\lambda_1$. For $w>N/2$, define
\[
\Gamma(w)
\triangleq
1-N^2+N+\frac{w(2N^2-4N+1)}{N},
\]
so that $\Gamma(w)$ is the slope of $r_{\mathrm{PAP}}(w;\lambda_1)$ with respect to $\lambda_1$.

From~\eqref{eq:PAP-low-support-affine-theta} and~\eqref{eq:PAP-high-support-affine-s}, the slopes of $r_{\mathrm{PAP}}(w_L;\lambda_1)$, $r_{\mathrm{PAP}}(w_H;\lambda_1)$, and $r_{\mathrm{PAP}}(N;\lambda_1)$ with respect to $\lambda_1$ are, respectively,
\[
1-w_L,\qquad \Gamma(w_H),\qquad \Gamma(N).
\]
For $N\geq3$,
\[
1-w_L\leq0,\qquad \Gamma(N)=(N-1)(N-2)>0.
\]
Thus, $r_{\mathrm{PAP}}(w_L;\lambda_1)$ is non-increasing and $r_{\mathrm{PAP}}(N;\lambda_1)$ is strictly increasing in $\lambda_1$.

The sign of $\Gamma(w_H)$ depends on whether $N$ is odd or even, so we consider these two cases separately.

\par\medskip
\noindent\underline{\textit{Odd $N$.}}\par\nopagebreak[4]

For odd $N$,
\[
w_L=\frac{N-1}{2},
\qquad
w_H=\frac{N+1}{2}.
\]
The preceding envelope reduction leaves $w_L$, $w_H$, and $N$ as the only candidate request sizes. We first show that the rate at $w_L$ never exceeds that at $w_H$. Indeed,
\[
g_N(w_L)=w_L,
\qquad
g_N(w_H)
=w_L+\frac{N+1}{2N}.
\]
Therefore, \eqref{PAP-R-alpha} gives, for every feasible pair $(\lambda_1,\lambda_2)$,
\[
\begin{aligned}
&r_{\mathrm{PAP}}(w_H;\lambda_1,\lambda_2)
-r_{\mathrm{PAP}}(w_L;\lambda_1,\lambda_2)\\
&\qquad=
\lambda_2\bigl(
g_N(w_H)-g_N(w_L)
\bigr)
+(1-\lambda_1-\lambda_2)(w_H-w_L)\\
&\qquad=
\lambda_2\frac{N+1}{2N}
+(1-\lambda_1-\lambda_2)
\geq0.
\end{aligned}
\]
Hence, after eliminating $\lambda_2$,
\[
\max_{w\in[N]}r_{\mathrm{PAP}}(w;\lambda_1)
=
\max\left\{
r_{\mathrm{PAP}}(w_H;\lambda_1),
r_{\mathrm{PAP}}(N;\lambda_1)
\right\}.
\]

The slopes of $r_{\mathrm{PAP}}(w_H;\lambda_1)$ and $r_{\mathrm{PAP}}(N;\lambda_1)$ with respect to $\lambda_1$ are
\[
\Gamma(w_H)
=
-\frac{N-1}{2N}
<0
\]
and
\[
\Gamma(N)
=
(N-1)(N-2)
>0,
\]
respectively. Thus, $r_{\mathrm{PAP}}(w_H;\lambda_1)$ is strictly decreasing in $\lambda_1$, whereas $r_{\mathrm{PAP}}(N;\lambda_1)$ is strictly increasing. Their unique intersection is equivalently characterized by $G(\lambda_1)=0$ and occurs at
\begin{equation}
\label{eq:PAP-odd-balancing-point}
\lambda_{1,\mathrm{bal}}
\triangleq
\frac{N\bigl((2N-1)M-1\bigr)}
{2N^2-4N+1}.
\end{equation}
The maximum of these two rates is minimized at $\lambda_{1,\mathrm{bal}}$ when the intersection lies in the feasible interval. If the intersection lies to the left or right of this interval, the minimum is attained at the corresponding endpoint.

\medskip
\noindent\textit{Case O1: $0\leq M\leq1/(2N-1)$.}

By \eqref{eq:PAP-theta-feasible-interval}, the lower endpoint of the feasible interval for $\lambda_1$ is
\[
\max\left\{
0,\frac{NM-1}{N-2}
\right\}.
\]
Since
\[
M\leq\frac{1}{2N-1}<\frac{1}{N},
\]
we have $NM-1<0$, so this lower endpoint equals $0$. Moreover, \eqref{eq:PAP-odd-balancing-point} gives
\[
\lambda_{1,\mathrm{bal}}\leq0.
\]
The intersection therefore lies at or to the left of the feasible interval.
For every feasible $\lambda_1$, \eqref{eq:PAP-theta-feasible-interval} gives
\[
\lambda_1\geq0\geq\lambda_{1,\mathrm{bal}}.
\]
Since $r_{\mathrm{PAP}}(N;\lambda_1)$ is increasing, $r_{\mathrm{PAP}}(w_H;\lambda_1)$ is decreasing, and the two rates are equal at $\lambda_{1,\mathrm{bal}}$, we have
\[
\begin{aligned}
r_{\mathrm{PAP}}(N;\lambda_1)
&\geq
r_{\mathrm{PAP}}(N;\lambda_{1,\mathrm{bal}})\\
&=
r_{\mathrm{PAP}}(w_H;\lambda_{1,\mathrm{bal}})\\
&\geq
r_{\mathrm{PAP}}(w_H;\lambda_1).
\end{aligned}
\]
Thus, $r_{\mathrm{PAP}}(N;\lambda_1)$ determines the upper envelope throughout the feasible interval. Since this rate is strictly increasing, the upper envelope is minimized at $\lambda_1^\star=0$. Substituting $\lambda_1^\star=0$ into \eqref{eq:PAP-active-memory} gives $\lambda_2^\star=NM$, and
\[
R_{\mathrm{PAP}}(N,M)
=
r_{\mathrm{PAP}}(N;0)
=
N-N(N-1)M.
\]

\medskip
\noindent\textit{Case O2: $1/(2N-1)<M\leq(N-1)/N$.}\par\nopagebreak[4]

In this range, \eqref{eq:PAP-odd-balancing-point} gives $\lambda_{1,\mathrm{bal}}>0$. By \eqref{eq:PAP-theta-feasible-interval}, the feasible interval is
\[
\max\left\{
0,\frac{NM-1}{N-2}
\right\}
\leq\lambda_1\leq
\frac{NM}{N-1}.
\]
Its upper endpoint is $NM/(N-1)$. The difference between this endpoint and $\lambda_{1,\mathrm{bal}}$ is
\[
\frac{NM}{N-1}-\lambda_{1,\mathrm{bal}}
=
\frac{N(N-1-NM)}
{(N-1)(2N^2-4N+1)}
\geq0.
\]
Hence, $\lambda_{1,\mathrm{bal}}$ does not exceed the upper endpoint.

It remains to compare $\lambda_{1,\mathrm{bal}}$ with the lower endpoint. If $M\leq1/N$, then $NM-1\leq0$, so the lower endpoint is $0$ and is strictly smaller than $\lambda_{1,\mathrm{bal}}$. If $M>1/N$, then the lower endpoint is $(NM-1)/(N-2)$, and
\[
\lambda_{1,\mathrm{bal}}
-\frac{NM-1}{N-2}
=
\frac{(N-1)(N-1-NM)}
{(N-2)(2N^2-4N+1)}
\geq0.
\]
Since $\lambda_{1,\mathrm{bal}}$ lies in the feasible interval throughout Case~O2, the upper envelope is minimized at
\[
\lambda_1^\star
=
\lambda_{1,\mathrm{bal}}
=
\frac{N\bigl((2N-1)M-1\bigr)}
{2N^2-4N+1}.
\]
Substituting this value into \eqref{eq:PAP-active-memory} gives
\[
\lambda_2^\star
=
\frac{N\bigl(N-1-NM\bigr)}
{2N^2-4N+1}.
\]
Therefore, the minimized worst-case rate is
\[
\begin{aligned}
R_{\mathrm{PAP}}(N,M)
&=
r_{\mathrm{PAP}}(w_H;\lambda_1^\star)\\
&=
r_{\mathrm{PAP}}(N;\lambda_1^\star)\\
&=
1+
\frac{(N-1)^2}
{2N^2-4N+1}
\bigl(N-1-NM\bigr).
\end{aligned}
\]
At $M=(N-1)/N$, the optimal pair becomes $(\lambda_1^\star,\lambda_2^\star)=(1,0)$ and the rate becomes $1$, consistent with the full-coordinate regime.

\par\medskip
\noindent\underline{\textit{Even $N$.}}\par\nopagebreak[4]

For even $N$, assume $N\geq4$. Then
\[
w_L=\frac{N}{2},
\qquad
w_H=\frac{N}{2}+1.
\]

By \eqref{eq:PAP-three-candidate-envelope}, only the rates corresponding to $w_L$, $w_H$, and $N$ can determine the upper envelope. Their dependence on $\lambda_1$ is particularly simple. For $w_L=N/2$,
\[
r_{\mathrm{PAP}}(w_L;\lambda_1)
=
\frac{N}{2}-\frac{N-2}{2}\lambda_1,
\]
which is strictly decreasing in $\lambda_1$. For the other two request sizes,
\[
\begin{gathered}\Gamma(w_H)
=
\frac{(2N-1)(N-2)}{2N}
>0,
\\
\Gamma(N)
=
(N-1)(N-2)
>0.\end{gathered}
\]
Hence, both $r_{\mathrm{PAP}}(w_H;\lambda_1)$ and $r_{\mathrm{PAP}}(N;\lambda_1)$ are strictly increasing in $\lambda_1$.

We first determine when $\lambda_1=0$ minimizes the upper envelope. By~\eqref{eq:PAP-theta-feasible-interval}, this choice is feasible precisely when $M\leq1/N$.

If one of the two increasing rates attains the maximum at $\lambda_1=0$, then increasing $\lambda_1$ increases that rate and cannot reduce the upper envelope. We therefore compare the three rates at $\lambda_1=0$.

Substituting $\lambda_1=0$ and $\lambda_2=NM$ into \eqref{PAP-R-alpha} gives
\begin{equation}
\label{eq:PAP-even-boundary-rates}
\begin{aligned}
r_{\mathrm{PAP}}(w_L;0)
&=
\frac{N}{2},\\
r_{\mathrm{PAP}}(w_H;0)
&=
\frac{N}{2}+1-\frac{3N-2}{2}M,\\
r_{\mathrm{PAP}}(N;0)
&=
N-N(N-1)M.
\end{aligned}
\end{equation}
To identify the largest of these rates, we compute
\begin{equation}
\label{eq:PAP-even-boundary-differences}
\begin{aligned}
r_{\mathrm{PAP}}(N;0)-r_{\mathrm{PAP}}(w_H;0)
&=
\frac{N-2}{2}\bigl[1-(2N-1)M\bigr],\\
r_{\mathrm{PAP}}(w_H;0)-r_{\mathrm{PAP}}(w_L;0)
&=
1-\frac{3N-2}{2}M.
\end{aligned}
\end{equation}
The two differences change sign at $M=1/(2N-1)$ and $M=2/(3N-2)$, respectively. Since $1/(2N-1)<2/(3N-2)<1/N$, both transitions occur while $\lambda_1=0$ is feasible. These two thresholds give the following three ranges of $M$.

\medskip
\noindent\textit{Case E1: $0\leq M\leq1/(2N-1)$.}

Both differences in \eqref{eq:PAP-even-boundary-differences} are nonnegative, so \eqref{eq:PAP-even-boundary-rates} gives
\[
r_{\mathrm{PAP}}(N;0)
\geq
r_{\mathrm{PAP}}(w_H;0)
\geq
r_{\mathrm{PAP}}(w_L;0).
\]
Thus, the upper envelope at $\lambda_1=0$ equals $r_{\mathrm{PAP}}(N;0)$. For every feasible $\lambda_1\geq0$, the monotonicity of $r_{\mathrm{PAP}}(N;\lambda_1)$ gives
\[
\max_{w\in[N]}r_{\mathrm{PAP}}(w;\lambda_1)
\geq
r_{\mathrm{PAP}}(N;\lambda_1)
\geq
r_{\mathrm{PAP}}(N;0).
\]
Therefore, the upper envelope is minimized at $\lambda_1^\star=0$. Equation~\eqref{eq:PAP-active-memory} then gives $\lambda_2^\star=NM$, and
\[
R_{\mathrm{PAP}}(N,M)
=
r_{\mathrm{PAP}}(N;0)
=
N-N(N-1)M.
\]

\medskip
\noindent\textit{Case E2: $1/(2N-1)<M\leq2/(3N-2)$.}

The first difference in \eqref{eq:PAP-even-boundary-differences} is negative, whereas the second is nonnegative. Hence,
\[
r_{\mathrm{PAP}}(w_H;0)
\geq
\max\left\{
r_{\mathrm{PAP}}(w_L;0),
r_{\mathrm{PAP}}(N;0)
\right\},
\]
so the upper envelope at $\lambda_1=0$ equals $r_{\mathrm{PAP}}(w_H;0)$. For every feasible $\lambda_1\geq0$, the monotonicity of $r_{\mathrm{PAP}}(w_H;\lambda_1)$ gives
\[
\max_{w\in[N]}r_{\mathrm{PAP}}(w;\lambda_1)
\geq
r_{\mathrm{PAP}}(w_H;\lambda_1)
\geq
r_{\mathrm{PAP}}(w_H;0).
\]
Therefore, the upper envelope is again minimized at $\lambda_1^\star=0$. By \eqref{eq:PAP-active-memory}, $\lambda_2^\star=NM$, and
\[
R_{\mathrm{PAP}}(N,M)
=
r_{\mathrm{PAP}}(w_H;0)
=
\frac{N}{2}+1-\frac{3N-2}{2}M.
\]

\medskip
\noindent\textit{Case E3: $2/(3N-2)<M\leq(N-1)/N$.}\par\nopagebreak[4]

For $2/(3N-2)<M\leq1/N$, \eqref{eq:PAP-even-boundary-differences} shows that $r_{\mathrm{PAP}}(w_L;0)$ strictly exceeds both increasing rates at $\lambda_1=0$. Thus, a sufficiently small feasible increase in $\lambda_1$ decreases the upper envelope. For $M>1/N$, $\lambda_1=0$ is infeasible. We therefore seek a minimizer with $\lambda_1>0$.

In this regime, the rate $r_{\mathrm{PAP}}(w_L;\lambda_1)$ decreases with $\lambda_1$, whereas $r_{\mathrm{PAP}}(w_H;\lambda_1)$ increases. Their unique intersection is
\begin{equation}
\label{eq:PAP-even-balancing-point}
\lambda_{1,\mathrm{bal}}^{\mathrm{(e)}}
\triangleq
\frac{N\bigl((3N-2)M-2\bigr)}{(3N-1)(N-2)}.
\end{equation}
We first verify that this intersection lies in the feasible interval \eqref{eq:PAP-theta-feasible-interval}. The lower bound on $M$ in Case~E3 gives $\lambda_{1,\mathrm{bal}}^{\mathrm{(e)}}>0$, and
\[
\frac{NM}{N-1}-\lambda_{1,\mathrm{bal}}^{\mathrm{(e)}}
=
\frac{2N(N-1-NM)}{(N-1)(3N-1)(N-2)}
\geq0,
\]
so the intersection does not exceed the upper endpoint. If $M\leq1/N$, the lower endpoint in \eqref{eq:PAP-theta-feasible-interval} is $0$ and is therefore smaller than $\lambda_{1,\mathrm{bal}}^{\mathrm{(e)}}$. If $M>1/N$, then the lower endpoint is $(NM-1)/(N-2)$, and
\[
\lambda_{1,\mathrm{bal}}^{\mathrm{(e)}}
-
\frac{NM-1}{N-2}
=
\frac{N-1-NM}{(3N-1)(N-2)}
\geq0.
\]
Thus, \eqref{eq:PAP-even-balancing-point} is feasible throughout Case~E3.

Since $r_{\mathrm{PAP}}(w_L;\lambda_1)$ decreases and $r_{\mathrm{PAP}}(w_H;\lambda_1)$ increases, their common value at the feasible intersection $\lambda_{1,\mathrm{bal}}^{\mathrm{(e)}}$ is a lower bound on the worst-case rate. To show that this lower bound is attained, it remains to verify that the third candidate rate, \(r_{\mathrm{PAP}}(N;\lambda_1)\), does not exceed this common value at the balancing point. Substituting \eqref{eq:PAP-even-balancing-point} into \eqref{eq:PAP-AG-definitions} gives
\[
G(\lambda_{1,\mathrm{bal}}^{\mathrm{(e)}})
=
-\frac{N(N-1-NM)}{(3N-1)(N-2)}
\leq0.
\]
By \eqref{eq:PAP-high-support-affine-s},
\[
\begin{aligned}
&r_{\mathrm{PAP}}(N;\lambda_{1,\mathrm{bal}}^{\mathrm{(e)}})
-r_{\mathrm{PAP}}(w_H;\lambda_{1,\mathrm{bal}}^{\mathrm{(e)}})\\
&\qquad=
(N-w_H)G(\lambda_{1,\mathrm{bal}}^{\mathrm{(e)}})
\leq0.
\end{aligned}
\]
Together with \eqref{eq:PAP-three-candidate-envelope}, this shows that the upper envelope at $\lambda_{1,\mathrm{bal}}^{\mathrm{(e)}}$ equals the common value of the rates for $w_L$ and $w_H$. The preceding lower bound therefore proves that
\[
\lambda_1^\star
=
\lambda_{1,\mathrm{bal}}^{\mathrm{(e)}}
=
\frac{N\bigl((3N-2)M-2\bigr)}{(3N-1)(N-2)}.
\]
Substituting this value into \eqref{eq:PAP-active-memory} gives
\[
\lambda_2^\star
=
\frac{2N\bigl(N-1-NM\bigr)}{(3N-1)(N-2)}.
\]
Finally,
\[
\begin{aligned}
R_{\mathrm{PAP}}(N,M)
&=
r_{\mathrm{PAP}}(w_L;\lambda_1^\star)\\
&=
r_{\mathrm{PAP}}(w_H;\lambda_1^\star)\\
&=
\frac{N(3N+1)-(3N-2)NM}{2(3N-1)}.
\end{aligned}
\]
At $M=(N-1)/N$, the optimal pair is $(\lambda_1^\star,\lambda_2^\star)=(1,0)$ and the rate is $1$, matching the full-coordinate regime.
\end{proof}
\section{Finite-Length Realization}
\label{app:finite-length}

Fix rational $\lambda_1,\lambda_2\geq0$ with $\lambda_1+\lambda_2\leq1$ and $(N-1)\lambda_1+\lambda_2\leq NM$. Choose $L$ so that $\lambda_1 L$, $\lambda_2 L$, $\lambda_1 L/N$, and $\lambda_2 L/N$ are integers. Partition $[L]$ into full, ABO, and uncached sets of the prescribed sizes. Assign exactly $\lambda_1 L/N$ full coordinates and $\lambda_2 L/N$ ABO coordinates to each source. Source $i$ caches the $N-1$ nonlocal symbols at each of its assigned full coordinates and one ABO parity at each of its assigned ABO coordinates. Its cache image has size
\[
2^{L[(N-1)\lambda_1+\lambda_2]/N}.
\]
The cached symbols at a full coordinate are independent, and an ABO parity is a nonzero linear form. Different coordinates involve disjoint file symbols, so the stated image size is exact. The delivery rules in Section~\ref{sec:pap} recover every requested coordinate. Each source's transmitted nonzero linear forms also involve disjoint symbols, so the logarithm of its message image size equals its transmitted bit count. Summing these counts gives~\eqref{PAP-R-alpha}.

For a real feasible pair $(\lambda_1,\lambda_2)$, choose rational sequences $0\leq\lambda_{1,k}\leq\lambda_1$ and $0\leq\lambda_{2,k}\leq\lambda_2$ converging to these parameters. Then $\lambda_{1,k}+\lambda_{2,k}\leq1$ and $(N-1)\lambda_{1,k}+\lambda_{2,k}\leq NM$, so every approximating pair is feasible. The maximum of the finitely many request rates in~\eqref{PAP-R-alpha} is continuous. The rational constructions therefore approach the rate of $(\lambda_1,\lambda_2)$ while respecting the same cache budget, proving asymptotic achievability.

\end{document}